\documentclass[a4paper,USenglish,cleveref,pdfa]{lipics-v2021}

\pdfoutput=1
\hideLIPIcs
\nolinenumbers

\usepackage{preamble}

\title{Parameterized Maximum Node-Disjoint Paths}

\titlerunning{Parameterized Maximum Node-Disjoint Paths}

\author{Michael Lampis}
{Universit\'{e} Paris-Dauphine, PSL University, CNRS UMR7243, LAMSADE, Paris, France}
{michail.lampis@dauphine.fr}
{https://orcid.org/0000-0002-5791-0887}{}

\author{Manolis Vasilakis}
{Universit\'{e} Paris-Dauphine, PSL University, CNRS UMR7243, LAMSADE, Paris, France}
{emmanouil.vasilakis@dauphine.eu}
{https://orcid.org/0000-0001-6505-2977}{}

\authorrunning{M. Lampis and M. Vasilakis}

\Copyright{Michael Lampis and Manolis Vasilakis}



\category{} 

\relatedversiondetails[linktext={https://arxiv.org/abs/2404.14849}]{Full Version}{https://arxiv.org/abs/2404.14849}

\funding{This work is partially supported by ANR project ANR-21-CE48-0022 (S-EX-AP-PE-AL).}

\EventEditors{Neeldhara Misra and Magnus Wahlstr\"{o}m}
\EventNoEds{2}
\EventLongTitle{18th International Symposium on Parameterized and Exact Computation (IPEC 2023)}
\EventShortTitle{IPEC 2023}
\EventAcronym{IPEC}
\EventYear{2023}
\EventDate{September 6--8, 2023}
\EventLocation{Amsterdam, The Netherlands}
\EventLogo{}
\SeriesVolume{285}
\ArticleNo{22}

\begin{document}

\maketitle

\begin{abstract}
We revisit the \textsc{Maximum Node-Disjoint Paths} problem,
the natural optimization version of the famous \textsc{Node-Disjoint Paths} problem, 
where we are given an undirected graph $G$, $k$ (demand) pairs of vertices $(s_i, t_i)$, and an integer $\ell$,
and are asked whether there exist at least $\ell$ vertex-disjoint paths in $G$ whose endpoints are given pairs.
This problem has been intensely studied from both the approximation and parameterized complexity point of view
and is notably known to be intractable by standard structural parameters,
such as tree-depth, as well as the combined parameter $\ell$ plus pathwidth.
We present several results improving and clarifying this state of the art, with an emphasis towards FPT approximation.

Our main positive contribution is to show that the problem's intractability can be overcome
using approximation: We show that for several of the structural parameters for which the problem is hard,
most notably tree-depth, the problem admits an \emph{efficient FPT approximation scheme}, returning a
$(1-\varepsilon)$-approximate solution in time $f(\mathrm{td},\varepsilon)n^{\mathcal{O}(1)}$.
We manage to obtain these results by comprehensively mapping out the structural parameters for which the problem
is FPT if $\ell$ is also a parameter, hence showing that understanding $\ell$ as a parameter is key to the problem's approximability. 
This, in turn, is a problem we are able to solve via a surprisingly simple color-coding algorithm,
which relies on identifying an insightful problem-specific variant of the natural parameter,
namely the number of vertices used in the solution.

The results above are quite encouraging, as they indicate that in some situations where the problem does not admit an FPT
algorithm, it is still solvable almost to optimality in FPT time.
A natural question is whether the FPT approximation algorithm we devised for tree-depth can be extended to pathwidth.
We resolve this negatively, showing that under the \emph{Parameterized Inapproximability Hypothesis} no FPT approximation
scheme for this parameter is possible, even in time $f(\mathrm{pw},\varepsilon)n^{g(\varepsilon)}$.
We thus precisely determine the parameter border where the problem transitions from
``hard but approximable'' to ``inapproximable''.

Lastly, we strengthen existing lower bounds by replacing
W[1]-hardness by XNLP-completeness for parameter pathwidth,
and improving the $n^{o(\sqrt{\mathrm{td}})}$ ETH-based lower bound
for tree-depth to (the optimal) $n^{o(\mathrm{td})}$.

\end{abstract}

\newpage

\section{Introduction}\label{sec:introduction}

One of the most important problems of structural graph theory has arguably been \NDP,
where given a graph $G$ and $k$ pairs of its vertices $(s_i, t_i)$ for $i = 1, \ldots, k$, called demands,
the goal is to determine whether there exist $k$ vertex-disjoint paths connecting $s_i$ and $t_i$.
This extensively studied problem~\cite{jct/AdlerKKLST17,iwpec/ChaudharyG0Z23,soda/Cho0O23,jct/KawarabayashiKR12,focs/KorhonenPS24,siamcomp/LokshtanovMS18,siamcomp/LokshtanovMP0Z25,scheffler1994practical,focs/WlodarczykZ23}
is one of the very first to be proven NP-complete (for $k$ being part of the input)~\cite{networks/Karp75},
and it has a central role in the field of structural graph theory as well as in parameterized complexity~\cite{birthday/Lokshtanov0Z20},
as the breakthrough result by Robertson and Seymour~\cite{jct/RobertsonS95b} that it is
fixed-parameter tractable (FPT)
parameterized by $k$ (i.e., admits an $f(k) n^{\bO(1)}$ algorithm for some function $f$, where $n = |V(G)|$)
is the culmination of their long and influential series of works on Graph Minors.

In this work we concern ourselves with \textsc{Maximum Node-Disjoint Paths} (\mNDP),
the natural generalization of {\NDP} where one asks whether at least $\ell \leq k$ demands can be routed by vertex-disjoint paths
(we say that a demand is routed when there exists a path connecting its endpoints in the set of
vertex-disjoint paths of the solution).
Notice that one could alternatively phrase this as an optimization problem and ask for the maximum number of demands that can be routed.
Even though {\mNDP} has been intensely studied with respect to its
approximability~\cite{stoc/ChuzhoyKL16,icalp/ChuzhoyKN18,toc/ChuzhoyKN21,siamcomp/ChuzhoyKN22,mp/KolliopoulosS04},
our understanding regarding its tractability under the perspective of parameterized complexity is rather limited.
Given the rich literature regarding {\NDP} and the importance of its structural parameterizations
(indeed, Scheffler's $\tw^{\bO(\tw)} n^{\bO(1)}$ algorithm~\cite{scheffler1994practical} is a key ingredient of the proof of~\cite{jct/RobertsonS95b}),
the quest to study {\mNDP} under the same point of view is strongly motivated,
with the hope of extending some of these results to it.
Alas, prior work by Ene, Mnich, Pilipczuk, and Risteski~\cite{swat/EneMPR16} shows that {\mNDP} is already \Wone-hard
when parameterized by the tree-depth of the input graph (in fact their proof implies hardness for the combined parameter vertex integrity%
\footnote{A graph has vertex integrity at most $k$ if there exists a set of at most $p$
vertices such that their deletion results in a graph with connected components of size at most $k - p$.}
plus feedback vertex number).
On the other hand, notice that {\mNDP} is trivially FPT by $k$;
one can simply reduce it to $2^k$ instances of \NDP.
A natural question arising from this observation is whether a parameterization by $\ell$ renders the problem tractable.
In this spirit, Marx and Wollan~\cite{soda/MarxW15} studied this setting and proved
that the problem is \Wone-hard even when parameterized by the combined parameter
$\ell$ plus the treewidth of the input graph;
a closer look into their proof reveals that their result extends to graphs of bounded
pathwidth plus feedback vertex number.
This plethora of negative results fails to answer which parameterizations render the problem tractable,
and whether a parameterization by $\ell$ plus some structural parameter (larger than or incomparable to treewidth)
may lift it to FPT.

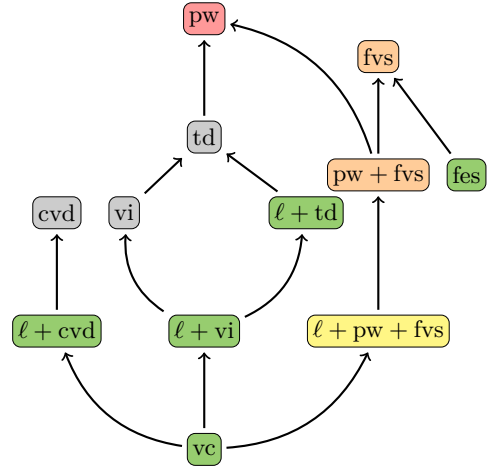
\begin{figure}[ht]
\centering
    \begin{tikzpicture}[every node/.style={thick, align=center}, scale=1]
        \small
        
        \node[fpt] (vc)   at (0,0) {$\vc$};
        \node[fpt, above = 1 of vc] (vil) {$\ell + \vi$};
        \node[fpt, above right = 1 and 0.25 of vil] (tdl) {$\ell + \td$};
        \node[fptas, above left = 1 and 0.25 of vil] (vi) {$\vi$};
        \node[fptas, above = 2 of vil] (td) {$\td$};
        \node[xnlp, above = 1 of td] (pw) {$\pw$};
        \draw[thick,<-] (vil) -- (vc);
        \draw[thick,<-] (tdl) edge [bend left] (vil);
        \draw[thick,<-] (vi) edge [bend right] (vil);
        \draw[thick,<-] (td) -- (vi);
        \draw[thick,<-] (td) -- (tdl);
        \draw[thick,<-] (pw) -- (td);
        
        \node[fpt, above left = 1 and 1 of vc] (cvdl) {$\ell + \cvd$};
        \node[fptas, above = 1 of cvdl] (cvd) {$\cvd$};
        \draw[thick,<-] (cvdl) edge [bend right] (vc);
        \draw[thick,<-] (cvd) -- (cvdl);
        
        \node[w1, above right = 1 and 1 of vc] (pwfvsl) {$\ell + \pw + \fvs$};
        \node[nph, above = 1.5 of pwfvsl] (pwfvs) {$\pw + \fvs$};
        \draw[thick,<-] (pwfvsl) edge [bend left] (vc);
        \draw[thick,<-] (pwfvs) -- (pwfvsl);
        \draw[thick,<-] (pw) edge [bend left] (pwfvs);
        
        \node[nph, above = 1 of pwfvs] (fvs) {$\fvs$};
        \node[fpt, below right = 1 and 0.5 of fvs] (fes) {$\fes$};
        \draw[thick,<-] (fvs) -- (fes);
        \draw[thick,<-] (fvs) -- (pwfvs);
    \end{tikzpicture}
    \caption{Our results and hierarchy of the related graph parameters,
        where $\vc$, $\vi$, $\td$, $\pw$, $\fvs$, $\cvd$, and $\fes$ stand for vertex cover,
        vertex integrity, tree-depth, pathwidth, feedback vertex number, cluster vertex deletion number,
        and feedback edge number respectively.
        For any graph, if the parameter at the tail of an arrow is a constant,
        that is also the case for the one at its head.
        Green and gray indicate that the problem is FPT (\cref{thm:fpt_algo,thm:fes}) and
        that it admits an FPT approximation scheme (\cref{thm:fptas_vi_cvd,thm:fptas_td}) respectively.
        Yellow indicates \Wone-hardness,
        orange that there is no FPT approximation scheme (\cref{thm:inapx_pw}),
        and red \XNLP-completeness (\cref{thm:pw_xnlp}).
        Prior to this work, it was only known that the problem is \Wone-hard parameterized
        by $\vi + \fvs$~\cite{swat/EneMPR16} and by $\ell + \pw + \fvs$~\cite{soda/MarxW15}.
    }
    \label{fig:parameters}
\end{figure}

\subparagraph*{Our Contribution.}
In the present paper we thoroughly investigate the complexity of {\mNDP} under different parameterizations,
and determine exactly when it is rendered tractable,
by additionally employing the use of approximation in the process
(see \cref{fig:parameters} for a synopsis of our results).
We start by showing that the problem is FPT parameterized by the number of vertices of an optimal solution
by developing a simple algorithm that makes use of the color-coding technique introduced by Alon, Yuster, and Zwick~\cite{jacm/AlonYZ95}.
We then prove that, albeit simple, this algorithm is in fact sufficient
to pinpoint exactly when the problem is fixed-parameter tractable:
utilizing a variety of structural observations,
we develop FPT algorithms for various parameterizations (most involving $\ell$ as a parameter)
at the core of all of which lies the previously mentioned algorithm.
Along the way we also develop an FPT algorithm for the parameterization by the feedback edge number.
These positive results, in conjunction with the hardness results of~\cite{swat/EneMPR16,soda/MarxW15},
clearly showcase the transition of the problem from FPT to \Wone-hard for its various parameterizations
(with the exception of cluster vertex deletion number, where it is unknown whether the problem is W[1]-hard).

Given the apparent hardness of the problem, we move on to consider it under the perspective of parameterized approximation (see~\cite{algorithms/FeldmannSLM20} for a survey of the area).
Here, we observe that utilizing the previously developed FPT algorithms when $\ell$ is also a parameter,
can in fact lead to efficient FPT approximation schemes (FPT-ASes) in case of solely structural parameterizations
of the problem,
when parameterized by the cluster vertex deletion number, the vertex integrity, or the tree-depth of the input graph;
in the latter two cases the problem is known to be \Wone-hard~\cite{swat/EneMPR16}.

The FPT approximation schemes developed indicate that the W[1]-hardness
can, in some cases, be overcome via the use of approximation.
Given the relationship of our studied parameters as well as the FPT-AS for tree-depth,
a natural question is whether an analogous approximation scheme exists for the parameterization by pathwidth as well.
Notice that the \Wone-hardness by $\ell$ plus the pathwidth $\pw$ and the feedback vertex number
$\fvs$ of the input graph already excludes the existence of an approximation scheme of running time
$f(\pw, \fvs, \varepsilon) n^{\bO(1)}$,
yet one of time $f(\pw, \fvs, \varepsilon) n^{g(\varepsilon)}$ remains possible.
Our next result is to exclude the existence of such a scheme
under the \emph{Parameterized Inapproximability Hypothesis}~\cite{soda/LokshtanovR0Z20},
which was recently proved to hold under the ETH~\cite{stoc/GuruswamiLRS024},
and thus precisely determine the parameter border where the problem transitions from
``hard but approximable'' to ``inapproximable''.
By slightly modifying our reduction, we subsequently show that the problem is \XNLP-complete
when parameterized solely by the pathwidth of the input graph,
where {\XNLP} is a complexity class that has been recently brought forth by Bodlaender, Groenland, Nederlof,
and Swennenhuis~\cite{iandc/BodlaenderGNS24}, and such a result implies W[$t$]-hardness for all integers $t \geq 1$.

Lastly, we proceed to a more fine-grained examination of the hardness of {\mNDP}
when parameterized by the tree-depth of the input graph.
Standard dynamic programming techniques can be used to obtain an $n^{\bO(\tw)}$ algorithm,
while previous work by Ene, Mnich, Pilipczuk, and Risteski~\cite{swat/EneMPR16} implies that
the problem cannot be solved in time $n^{o(\sqrt{\td})}$ under the ETH for graphs of tree-depth $\td$,
thereby leaving hope for an $n^{o(\td)}$ algorithm.
We revisit said proof, and by employing a recursive structure introduced by Lampis and Vasilakis~\cite{toct/LampisV24}
we bridge this gap and prove that {\mNDP} cannot be solved in time $n^{o(\td)}$ under the ETH,
rendering the $n^{\bO(\tw)}$ algorithm optimal even for this much smaller class of graphs.

\subparagraph*{Related Work.}
Even though {\mNDP} has been well-studied under the scope of approximation algorithms,
the 20 years old algorithm of ratio $\bO(\sqrt{n})$
due to Kolliopoulos and Stein~\cite{mp/KolliopoulosS04} remains the state of the art in general graphs.
This has been improved in the case of grid~\cite{approx/ChuzhoyK15} and planar graphs~\cite{stoc/ChuzhoyKL16},
resulting in approximation ratios $\tO(n^{1/4})$ and $\tO(n^{9/19})$ respectively,
where standard~$\tO$ notation is used to hide polylogarithmic terms.
For graphs of pathwidth~$\pw$, Ene, Mnich, Pilipczuk, and Risteski~\cite{swat/EneMPR16}
have presented an algorithm of approximation ratio $\bO(\pw^3)$.
Regarding inapproximalibity results, after a series of works, Chuzhoy, Kim, and Nimavat~\cite{toc/ChuzhoyKN21,siamcomp/ChuzhoyKN22}
have shown that the problem cannot be approximated in polynomial time
(i) within a factor of~$2^{\bO(\log^{1 - \varepsilon} n)}$ for any constant $\varepsilon$,
assuming $\NP \not\subseteq \text{DTIME}(n^{\mathrm{polylog} n})$,
and (ii) within a factor of~$n^{\bO(1 / (\log \log n)^2)}$,
assuming that there exists some constant $\delta > 0$ such that $\NP \not\subseteq \text{DTIME}(2^{n^\delta})$.

The problem has been also studied under a parameterized complexity perspective.
As already noted, it is trivially FPT by~$k$ by a simple reduction to~$2^k$ instances of \NDP,
which is well-known to be FPT by the number of demands.
On the other hand, Marx and Wollan~\cite{soda/MarxW15} have shown that it becomes \Wone-hard when
parameterized by~$\ell + \tw$,
with a closer look into their proof revealing that their result extends to the parameterization by~$\ell + \pw + \fvs$.
Regarding structural parameterizations,
Ene, Mnich, Pilipczuk, and Risteski~\cite{swat/EneMPR16} have proved that the problem is \Wone-hard when parameterized by
the tree-depth of the input graph;
in fact, their proof extends to graphs of bounded vertex integrity plus feedback vertex number.
Fleszar, Mnich, and Spoerhase~\cite{mp/FleszarMS18} have proposed an algorithm of running time~$(k+\fvs)^{\bO(\fvs)} n^{\bO(1)}$
as an alternative to the~$2^k \fvs^{\bO(\fvs)} n^{\bO(1)}$ algorithm obtained by reducing the instance to {\NDP} and then using Scheffler's algorithm~\cite{scheffler1994practical}.

\subparagraph*{Organization.}
In \cref{sec:preliminaries} we discuss the general preliminaries.
Subsequently, in \cref{sec:tractability} we present various tractability results,
followed by the inapproximability result for pathwidth in \cref{sec:inapx}.
Moving on, in \cref{sec:xnlp,sec:td_lb} we present the XNLP-completeness and the refined \Wone-hardness of the problem, when parameterized by the pathwidth and the tree-depth of the input graph respectively.
Lastly, in \cref{sec:conclusion} we present the conclusion as well as some directions for future research.
Proofs of statements marked with {\appsymbNote} are deferred to the appendix.

\section{Preliminaries}\label{sec:preliminaries}
For $x, y \in \Z$, let $[x, y] = \setdef{z \in \Z}{x \leq z \leq y}$,
while $[x] = [1,x]$.
For a set $S$, let $\binom{S}{c}$ denote the set of subsets of $S$ of size $c$ for some $c \in \N$,
that is $\binom{S}{c} = \setdef{S' \subseteq S}{|S'| = c}$.
Throughout the paper we use standard graph notations~\cite{Diestel17}
and assume familiarity with the basic notions of parameterized complexity~\cite{books/CyganFKLMPPS15}.
All graphs considered are undirected without loops unless explicitly mentioned otherwise.
Let $G = (V,E)$ be a graph.
The \emph{cluster vertex deletion number} of $G$, denoted~$\cvd(G)$, is the size of the smallest vertex set
whose removal results in a cluster graph, i.e., a union of cliques.
The \emph{vertex integrity} of $G$, denoted~$\vi(G)$, is the minimum integer~$k$
such that there is a vertex set $S \subseteq V$ such that $|S| + \max_{C \in \cc(G-S)} |V(C)| \le k$,
where~$\cc(G-S)$ denotes the set of connected components in $G-S$.

Given a graph $G = (V, E)$ and a partition of $V$ into $k$ independent sets $V_1, \ldots, V_k$, each of size $n$,
{\kMC} asks whether $G$ contains a $k$-clique,
and is well-known to be \Wone-hard and not to admit any $f(k) n^{o(k)}$ algorithm,
where $f$ is any computable function,
unless the ETH is false~\cite{books/CyganFKLMPPS15}.
{\MDkS} asks for the maximum number of edges that are induced by a subgraph of $G$ that
contains exactly one vertex per independent set (color).
Let $E^{i,j} \subseteq E$ denote the set of edges $e = \braces{u,v}$ where $u \in V_i$ and $v \in V_j$.
In that case, let $s_i = |\setdef{j \in [k]}{E^{i,j} \neq \varnothing}|$ and $\sigma = \sum_{i \in [k]} s_i / 2$.
Assuming that $\OPT$ denotes said maximum value, the \emph{Parameterized Inapproximability Hypothesis} (PIH)~\cite{soda/LokshtanovR0Z20} states that
there exists a constant $0 < c < 1$ such that no $f(k) n^{\bO(1)}$ algorithm can distinguish
between $\OPT = \sigma$ and $\OPT < c \cdot \sigma$.
In fact, one can assume without loss of generality that $1 \leq s_i \leq 3$
for all $i = 1, \ldots, k$~\cite[Lemma~4.4]{soda/LokshtanovR0Z20}.
This hypothesis was very recently proved to hold under the ETH~\cite{stoc/GuruswamiLRS024} (see also~\cite{stoc/GuruswamiLRS025}),
and is the analogue of the PCP theorem in the setting of parameterized complexity.

Lastly, we give a formal definition of the problem this work is concerned with.

\problemdef{Maximum Node-Disjoint Paths}
{Graph $G = (V,E)$, set of $k$ demand pairs $\M \subseteq \binom{V}{2}$, and integer $\ell \leq k$.}
{Determine whether at least $\ell$ demand pairs can be routed,
where to route a pair we need to select a path connecting it,
so that all selected paths are vertex-disjoint.}
Notice that in the above definition, even though demand pairs may indeed share a terminal,
the paths comprising a feasible solution must be vertex-disjoint,
and this constraint also applies to their endpoints.
Given an instance $\mathcal{I} = (G, \M (, \ell))$ of the optimization (or decision) version of {\mNDP}
and $G'$ a subgraph of $G$,
$\OPT(\mathcal{I}[G'])$ denotes the maximum number of demands of $\M$ that can be routed in $G'$.
We write $\OPT(\mathcal{I})$ as a shorthand for $\OPT(\mathcal{I}[G])$.

\section{FPT Algorithms and Approximation Schemes}\label{sec:tractability}

Here we present various tractability results for \mNDP.
We start by proving in \cref{subsec:exact_fpt} that the problem is FPT when parameterized by the number of vertices involved in an optimal solution;
using this as well as some structural observations,
we obtain FPT algorithms when parameterized by $\vc$, $\ell + \cvd$, $\ell + \vi$, and $\ell + \td$.
We additionally develop an FPT algorithm for the parameterization by $\fes$.
Moving on, in \cref{subsec:fpt_as} we obtain FPT approximation schemes for various structural
parameterizations of the problem by making use of the previous FPT algorithms;
for most of said parameterizations the problem is known to be \Wone-hard.

\subsection{Exact Algorithms}\label{subsec:exact_fpt}

We start with the following theorem.

\begin{theorem}\label{thm:number_of_vertices}
    Let $\mathcal{I} = (G, \M, \ell)$ be an instance of \mNDP.
    Additionally, let $\tau$ be such that there exists a family $\mathcal{P}$ of vertex-disjoint
    paths each routing a demand of $\M$,
    where $|\mathcal{P}| = \min \braces{\ell, \OPT(\mathcal{I})}$
    and $\sum_{P \in \mathcal{P}} |V(P)| \leq \tau$,
    with $|V(P)|$ denoting the number of vertices in path $P$.
    There is an algorithm that, given $\mathcal{I}$ and $\tau$,
    decides $\mathcal{I}$ in time $2^{\bO(\tau)} n^{\bO(1)}$.
\end{theorem}

\begin{proof}
    Let $C = [\tau]$ be a set of $\tau$ colors.
    Randomly color the vertices of $G$ with colors from~$C$,
    and let~$A$ be the event where every one of the at most $\tau$ vertices of $\mathcal{P}$ receives a distinct color.
    Then, it follows that
    \[
        \Pr[A] \geq \frac{\tau!}{\tau^\tau} > \frac{\parens*{\frac{\tau}{e}}^\tau}{\tau^\tau} = e^{-\tau},
    \]
    therefore event $A$ holds with probability at least $e^{-\tau}$.

    Now, let $T[S]$ be equal to the maximum number of demands that can be routed by paths
    using only vertices of colors belonging to $S \subseteq C$, where each color is used in at most one path.
    Moreover, let $f(S) = 1$ if there exists at least one demand that can be routed using only vertices of colors belonging to $S$ and $0$ otherwise.
    Notice that $f$ can be computed in polynomial time.
    Then, it holds that
    \[
        T[S] = \max_{S' \subseteq S} \big\{f(S') + T[S \setminus S'] \big\},
    \]
    thus one can compute $T[C]$ in time $2^{\bO(\tau)} n^{\bO(1)}$ for a given coloring of the vertices of $G$.

    By repeating this procedure $2^{\bO(\tau)}$ times,
    with high probability there exists some iteration where event $A$ holds,
    and the total running time is $2^{\bO(\tau)} n^{\bO(1)}$.
    By using standard techniques~\cite[Section~5.6]{books/CyganFKLMPPS15},
    one can derandomize the described algorithm and obtain a deterministic one of the same running time.
\end{proof}

Using \cref{thm:number_of_vertices} we can obtain various parameterized algorithms,
by bounding the number of vertices of an optimal solution using some simple observations.

\begin{theoremrep}[\appsymb]\label{thm:fpt_algo}
    Given an instance $\mathcal{I} = (G, \M, \ell)$ of \mNDP,
    there exist algorithms that decide $\mathcal{I}$ in time
    \begin{itemize}
        \item $2^{\bO(\vc)} n^{\bO(1)}$,
        \item $2^{\bO(\cvd + \ell)} n^{\bO(1)}$,
        \item $2^{\bO(\vi^2 + \vi \cdot \ell)} n^{\bO(1)}$,
        \item $2^{\bO(2^\td \cdot \ell)} n^{\bO(1)}$,
    \end{itemize}
    where $\vc$, $\cvd$, $\vi$, and $\td$ denote the vertex cover,
    cluster vertex deletion number,
    vertex integrity, and tree-depth of $G$ respectively.
\end{theoremrep}

\begin{proof}
    We prove the statement by providing bounds on the number of vertices involved in an optimal solution,
    denoted by $\tau$,
    and then using the algorithm of \cref{thm:number_of_vertices}.
    Fix an optimal solution $\mathcal{P}$, comprised of paths $P_i$ with endpoints $s_i$ and $t_i$, for $i \in [r]$, where $0 \leq r \leq \ell$
    and $(s_i,t_i) \in \M$.
    We will denote the deletion set in each case by $S \subseteq V(G)$.
    Note that there is no need to actually compute said deletion set.

    In the case of the vertex cover, notice that at least one endpoint of every edge in $P_i \in \mathcal{P}$ belongs to $S$,
    while each vertex is involved in at most $2$ edges.
    In that case, it follows that $\tau \leq 3\vc$.

    We next consider the case of cluster vertex deletion number.
    For every path $P_i$,
    either both $s_i$ and $t_i$ belong to the same clique of $G-S$ or not.
    In the first case, there exists an optimal solution that considers the path on vertex set $\braces{s_i, t_i}$ instead.
    In the latter, every such path involves at least one vertex of $S$, for a total of at most $\cvd$ such paths.
    Moreover, if such a path involves more than $2$ vertices of the same clique of $G-S$, say $u_1, u_2, \ldots, u_q$,
    indexed by their order of appearance in the path,
    there exists an optimal solution which is obtained by taking the edge between $u_1$ and $u_q$ instead.
    Since every vertex of $S$ in such a path might be neighbors with vertices belonging to at most $2$ different cliques of $G-S$,
    it follows that after ``short-cutting'' all such paths, at most $\cvd + 4\cvd = 5\cvd$ vertices are used.
    In total, it follows that $\tau \leq 5\cvd + 2\ell$.

    We then consider the case of vertex integrity.
    For every path $P_i$,
    either $V(P_i) \cap S = \varnothing$ or $V(P_i) \cap S \neq \varnothing$.
    In the first case, it follows that such a path contains at most $\vi$ vertices.
    In the latter, $P_i$ has at least one vertex of $S$, for a total of at most $\vi$ such paths.
    Moreover, since every vertex of $S$ in $P_i$ might be neighbors with vertices belonging to at most $2$ different connected components of
    $G-S$, at most $\vi + \vi (2\vi) = 2\vi^2 + \vi$ vertices are used.
    In total, it follows that $\tau \leq 2\vi^2 + \vi + \vi \cdot \ell$.

    For tree-depth, it holds that any path in $G$ has length at most $2^\td$,
    therefore $\tau \leq 2^\td \cdot \ell$.
\end{proof}

Given the \Wone-hardness of {\mNDP} when parameterized
by the feedback vertex number implied by previous works~\cite{swat/EneMPR16,soda/MarxW15},
we move on to consider the parameterization by the feedback edge number,
and show that it renders the problem tractable.

\begin{theoremrep}[\appsymb]\label{thm:fes}
    Given an instance $\mathcal{I} = (G, \M)$ of \mNDP,
    there exists an algorithm that computes $\OPT(\mathcal{I})$ in time $3^{\fes} n^{\bO(1)}$,
    where $\fes$ denotes the feedback edge number of $G$.
\end{theoremrep}

\begin{proof}
    The algorithm will perform branching and reduce the instance to a collection of cycles,
    in which case the problem is polynomial-time solvable.
    We first present some reduction rules where Rules 2 and 3 are only applied after exhaustively applying Rule 1.

    \proofsubparagraph*{Rule 1.}
    Let $\mathcal{I} = (G, \M)$ be an instance of \mNDP,
    and $u \in V(G)$ such that $(u,u) \in \M$.
    Then, replace $\mathcal{I}$ with $\mathcal{I}' = (G - u, \M')$,
    where $\M' \subseteq \M$ contains the demands of $\M$ whose endpoints both differ from $u$.
    It holds that $\OPT(\mathcal{I}) = \OPT(\mathcal{I}') + 1$.

    \proofsubparagraph*{Rule 2.}
    Let $\mathcal{I} = (G, \M)$ be an instance of \mNDP,
    and $u \in V(G)$ such that $\deg_G(u) = 0$.
    Then, replace $\mathcal{I}$ with $\mathcal{I}' = (G - u, \M')$,
    where $\M' \subseteq \M$ contains the demands of $\M$ whose endpoints both differ from $u$.
    It holds that $\OPT(\mathcal{I}) = \OPT(\mathcal{I}')$.

    \proofsubparagraph*{Rule 3.}
    Let $\mathcal{I} = (G, \M)$ be an instance of \mNDP,
    and $u \in V(G)$ such that $\deg_G(u) = 1$,
    where $v$ denotes its single neighbor, i.e., $N_G(u) = \braces{v}$.
    Then, replace $\mathcal{I}$ with $\mathcal{I}' = (G - u, \M')$,
    where $\M'$ is obtained by replacing any demand $(u,w)$ in $\M$ with $(v,w)$.
    It holds that $\OPT(\mathcal{I}) = \OPT(\mathcal{I}')$.

    It is easy to see that applying these rules in their respective order is safe,
    and let $\mathcal{I} = (G, \M)$ denote the instance obtained after exhaustively doing so.
    Assume without loss of generality that $G$ is connected, otherwise solve each connected component independently.
    Notice that all vertices of $G$ have degree at least 2.
    Recall that for $G'$ subgraph of $G$ it holds that $\fes(G') \leq \fes(G)$,
    and since $G$ is a subgraph of the initial graph $\fes(G) \leq \fes$ follows.
    Moreover, for a connected graph $G$, it holds that $\fes(G) = |E(G)| - |V(G)| + 1$.

    Let $\mathcal{P}$ be a maximum-cardinality set of vertex-disjoint paths routing demands of $\M$ in $G$.
    Let $v \in V(G)$ such that $\deg_G(v) \geq 3$, where $e_1, e_2, e_3 \in E(G)$ denote three edges incident with $v$.
    Notice that at least one among those edges, say $e_1$, does not take part in $\mathcal{P}$,
    that is, $e_1 \notin E(P)$ for all paths $P \in \mathcal{P}$.
    Perform branching on all 3 cases and delete the corresponding edge from the graph of the produced instance;
    we claim that this reduces the feedback edge number of all connected components of the resulting graph.
    To see this, consider two cases: for $j \in [3]$, either $G - e_j$ is connected or not;
    notice that in the latter case we can solve for each connected component independently.
    If $G - e_j$ is connected, then it holds that $\fes(G - e_j) = \fes(G) - 1$.
    Otherwise, $G - e_j$ is comprised of two connected components $C_1$ and $C_2$,
    in which case $\fes(C_1) + \fes(C_2) = (|E(C_1)| + |E(C_2)|) - (|V(C_1)| + |V(C_2)|) + 2 =
    (|E(G)| - 1) - |V(G)| + 2 = \fes(G)$.
    We next prove that $\fes(C_i) > 0$ for $i \in [2]$.
    Assume that this is not the case, and let without loss of generality $\fes(C_1) = 0$,
    in which case $C_1$ is a tree.
    Moreover, let $w$ denote the vertex that is incident with $e_j$ and belongs to $C_1$.
    Since no reduction rule can be further applied in $G$,
    it follows that $\deg_{G[C_1]} (u) > 2$ for all vertices $u \neq w$ belonging to $C_1$.
    In that case, $G[C_1]$ has at most one leaf, consequently it is a singleton.
    Then however, $\deg_G (w) = 1$ and the reduction rules could have been further applied in $G$, a contradiction.
    Notice that this additionally implies that if $\fes(G) = 1$,
    then $G - e_j$ is connected.
    From the previous discussion, it follows that
    \begin{equation}\label{eq:fes}
        T(z) \leq 3 \max_{\substack{z_1, z_2 > 0 \\ z_1+z_2=z}} \big\{ T(z_1)+T(z_2) \big\},
    \end{equation}
    where $T(z)$ denotes the number of connected components resulting from
    our algorithm for $z \in [\fes]$ being the feedback edge number of the connected input graph,
    with each component having maximum degree $2$.

    We first argue that when the graph has maximum degree $2$, then the problem is polynomial-time solvable.
    Let $\mathcal{I}^\star = (G^\star, \M^\star)$ denote one such instance,
    where every vertex of the connected graph $G^\star$ is of degree exactly 2,
    as otherwise the reduction rules can still be applied.
    Consequently, $G^\star$ is a cycle, thus it suffices to guess the endpoints of one routed demand
    and delete the involved vertices
    (since $G^\star$ is a cycle there are only two different such paths),
    thus reducing $G^\star$ into a path where the reduction rules can be applied exhaustively.
    Therefore, one can compute $\OPT(\mathcal{I}^\star)$ in polynomial time.
    
    Lastly, we prove by induction that $T(z) \leq 3 (3^z - 9/5)$ for all $z \in [\fes]$.
    For the base case, $T(1) \leq 3$ holds.
    Assume that the statement holds for all $z' \in [z-1]$.
    Then, by \cref{eq:fes} it holds that
    $T(z) \leq 9 (3^{z_1} + 3^{z_2} - 18/5)
    \leq 9 (3^{z_1+z_2-1} + 3^{1} - 18/5)
    = 9 (3^{z-1} - 3/5)
    = 3 (3^z - 9/5)$,
    where the first inequality is due to the induction hypothesis
    whereas the second due to the function $3^x$ being convex.
\end{proof}

\subsection{Approximation Schemes}\label{subsec:fpt_as}

Using the FPT algorithms of \cref{thm:fpt_algo},
we develop FPT approximation schemes for {\mNDP} when parameterized solely by structural parameters.

\begin{theorem}\label{thm:fptas_vi_cvd}
    Given an instance $\mathcal{I} = (G, \M)$ of \mNDP,
    one can $(1-\varepsilon)$-approximate $\OPT(\mathcal{I})$ in time
    $2^{\bO(\cvd / \varepsilon)} n^{\bO(1)}$ and $2^{\bO(\vi^2 / \varepsilon)} n^{\bO(1)}$,
    where $\cvd$ and $\vi$ denote the cluster vertex deletion number and
    vertex integrity of $G$ respectively.
\end{theorem}

\begin{proof}
    Let $S \subseteq V(G)$ denote the deletion set,
    which can be computed in time $2^{\bO(\cvd)} n^{\bO(1)}$ for $\cvd$~\cite{mst/Tsur21}
    and $2^{\bO(\vi \log \vi)} n^{\bO(1)}$ for $\vi$~\cite{algorithmica/DrangeDH16}.
    Notice that $\OPT(\mathcal{I}[G-S]) \geq \OPT(\mathcal{I}) - |S|$,
    since every vertex of $S$ can be used to route at most one demand.
    Consider the case where $\OPT(\mathcal{I}) - |S| \geq (1-\varepsilon) \OPT(\mathcal{I}) \iff \OPT(\mathcal{I}) \geq |S| / \varepsilon$.
    Then, $\OPT(\mathcal{I}[G-S]) \geq (1-\varepsilon) \OPT(\mathcal{I})$.
    Alternatively, it holds that $\OPT(\mathcal{I}) < |S| / \varepsilon$ and the algorithm of \cref{thm:fpt_algo} can be used,
    with $\ell = \bO(|S| / \varepsilon)$.

    It remains to compute $\OPT(\mathcal{I}[G-S])$.
    In the case of the cluster vertex deletion set, $G-S$ is a collection of cliques.
    In that case, one can compute $\OPT(\mathcal{I}[G-S])$ in polynomial time
    by reducing to an instance of \textsc{Maximum Matching},
    on the same vertex set and on edge set equal to the pairs of $\M$ where both endpoints belong to the same clique.
    In the case of vertex integrity, it holds that $\OPT(\mathcal{I}[G-S])$ can be computed in FPT time,
    since it is comprised of $\bO(n)$ instances of size at most $\vi$, each of which is solvable in time $2^{\bO(\vi)} n^{\bO(1)}$ due to \cref{thm:number_of_vertices}.
\end{proof}

\begin{theorem}\label{thm:fptas_td}
    Given an instance $\mathcal{I} = (G, \M)$ of \mNDP,
    one can $(1-\varepsilon)$-approximate $\OPT(\mathcal{I})$ in time
    $2^{\bO(2^\td / \varepsilon)} n^{\bO(1)}$,
    where $\td$ denotes the tree-depth of $G$.
\end{theorem}

\begin{proof}
    Since $G$ has tree-depth $\td$,
    one can compute an \emph{elimination forest} of $G$ in time $2^{\bO(\td^2)} n^{\bO(1)}$
    due to~\cite{esa/NadaraPS22,icalp/ReidlRVS14}.
    This is a rooted forest on the same vertex set where every pair of vertices adjacent in $G$
    adheres to the ancestor/descendant relation.
    Assume that~$G$ is connected, otherwise solve for each connected component.
    Let $r$ denote the root of the elimination tree,
    in which case it follows that every connected component of $G-r$ has tree-depth at most $\td-1$.
    Notice that it holds that $\OPT(\mathcal{I}[G-r]) \geq \OPT(\mathcal{I}) - 1$,
    since $r$ can be used to route at most one demand.
    Let $\varepsilon' < \varepsilon$ and consider the case where
    \[
        \OPT(\mathcal{I}) - 1 \geq \frac{1-\varepsilon}{1-\varepsilon'} \cdot \OPT(\mathcal{I}) \iff
        \OPT(\mathcal{I}) \geq \frac{1-\varepsilon'}{\varepsilon-\varepsilon'},
    \]
    thus setting $\varepsilon' = \varepsilon / 2$ implies that
    $\OPT(\mathcal{I}) \geq \frac{2-\varepsilon}{\varepsilon}$.
    In this case, it holds that an $(1-\varepsilon')$-approximation of $\OPT(\mathcal{I}[G-r])$ is an $(1-\varepsilon)$-approximation of $\OPT(\mathcal{I})$.
    On the other hand, if $\OPT(\mathcal{I}) < \frac{2-\varepsilon}{\varepsilon} = \bO(1 / \varepsilon)$,
    we can use the algorithm of \cref{thm:fpt_algo}, running in time $2^{\bO(2^\td / \varepsilon)} n^{\bO(1)}$.

    Consequently, one can recursively argue about the existence of an approximation scheme,
    since in the case where a graph has tree-depth equal to $1$,
    the instance is polynomial-time solvable.
    We proceed with bounding the scheme's running time.
    Let $T(\td,\varepsilon)$ denote the running time for graph $G$ with tree-depth $\td$ and error $\varepsilon$,
    while $n'$ denotes the number of connected components of $G-r$.
    Notice that it holds that
    \[        
        T(\td,\varepsilon) \leq \max\braces*{2^{\bO(2^\td / \varepsilon)} n^{\bO(1)}, n' \cdot T(\td-1, \varepsilon/2)}
        \leq 2^{\bO(2^\td / \varepsilon)} n^{\bO(1)} + n' \cdot T(\td-1, \varepsilon/2),
    \]
    while the number of the nodes of the same height in the recursion tree is at most $n$,
    since each node corresponds to a connected component.
    Consequently, it holds that
    \[
        T(\td,\varepsilon) \leq 2^{\bO(2^\td / \varepsilon)} n^{\bO(1)} +
            n \cdot \sum_{i=1}^{\td} 2^{\bO(2^{\td-i} \cdot \frac{1}{\varepsilon / 2^i})} n^{\bO(1)} =
        2^{\bO(2^\td / \varepsilon)} n^{\bO(1)}. \qedhere
    \]
\end{proof}

\section{Inapproximability}\label{sec:inapx}

Given the FPT approximation scheme of \cref{thm:fptas_td},
a natural question arising is whether such an approximation scheme exists for the parameterization by pathwidth as well.
Due to the \Wone-hardness of {\mNDP} parameterized by $\ell + \pw + \fvs$~\cite{soda/MarxW15},
there can be no $(1-\varepsilon)$-approximation scheme running in time $f(\pw,\fvs,\varepsilon) n^{\bO(1)}$,
yet one of running time $f(\pw,\fvs,\varepsilon) n^{g(\varepsilon)}$ might be possible.%
\footnote{Assuming there existed such an algorithm of running time $f(\pw,\fvs,\varepsilon) n^{\bO(1)}$,
setting $\varepsilon < 1 / \OPT(\mathcal{I})$
results in obtaining a solution of value at least $(1-\varepsilon) \OPT(\mathcal{I}) > \OPT(\mathcal{I}) - 1$,
i.e., optimal,
in time $f(\pw, \fvs, \OPT(\mathcal{I})) n^{\bO(1)}$.}
In this section we answer this question in the negative and prove that there exists some constant
$0 < c' < 1$ such that {\mNDP} cannot be approximated within a factor of $c'$ in time $f(\pw,\fvs) n^{\bO(1)}$,
for any function $f$, unless the Parameterized Inapproximability Hypothesis~\cite{soda/LokshtanovR0Z20} fails.

\begin{theorem}\label{thm:inapx_pw}
    Assuming the PIH, there exists a constant $c' > 0$ such that
    {\mNDP} does not admit a $c'$-approximation algorithm of running time $f(\pw,\fvs) n^{\bO(1)}$.
\end{theorem}

\begin{proof}
    Let $\mathcal{I} = (G,k)$ be an instance of \MDkS,
    and recall that we assume that $G$ is given to us partitioned into $k$ independent sets $V_1, \ldots, V_k$,
    where $V_i = \braces{v^i_1, \ldots, v^i_n}$.
    Moreover, let $E^{i_1,i_2} \subseteq E(G)$ denote the edges of $G$ with one endpoint in $V_{i_1}$ and the other in $V_{i_2}$.
    For every color class $i \in [k]$, let $s_i = |\setdef{j \in [k]}{E^{i,j} \neq \varnothing}|$,
    and assume without loss of generality that $1 \leq s_i \leq 3$.
    Set $\sigma = \sum_{i \in [k]} s_i / 2$,
    and notice that $\frac{k}{2} \leq \sigma \leq \frac{3k}{2}$.
    Let $\OPT(\mathcal{I}) \leq \sigma$ denote the optimal value of instance $\mathcal{I}$,
    i.e., the maximum number of edges among the induced
    subgraphs of $G$ that contain one vertex per color class~$V_i$.

    We will construct in polynomial time an instance $\mathcal{J} = (H, \M)$ of \mNDP,
    where $\pw(H) = \bO(k)$, $\fvs(H) = \bO(k)$, and $|V(H)| = n^{\bO(1)}$,
    while $\M \subseteq \binom{V(H)}{2}$ is a set of demands.
    Moreover, it will hold that $\OPT(\mathcal{J}) \leq \ell$,
    where $\OPT(\mathcal{J})$ denotes the optimal value of instance $\mathcal{J}$,
    and $\ell = 2k + \sigma$.
    We will present a reduction such that
    (i) if $\OPT(\mathcal{I}) = \sigma$, then $\OPT(\mathcal{J}) = \ell$,
    and (ii) if $\OPT(\mathcal{I}) < c \cdot \sigma$, then $\OPT(\mathcal{J}) < c' \cdot \ell$,
    for constants $c$ and $c'$ where
    $0 < c < 1$ and $c' = \frac{9+c}{10}$.

    \proofsubparagraph*{Choice Gadget.}
    For an independent set $V_i$,
    we construct the \emph{choice gadget} $\hat{C}_{i}$ in the following way.
    First, for $p \in [n]$, we construct paths on vertex sets $\hat{P}^{i}_p = \setdef{w^{i,p}_q}{q \in [3]}$,
    as well as paths $\hat{R}^{i}_p$ on $3$ unnamed vertices each.
    Then, we introduce vertices $v^{i}_p$ and $u^{i}_p$, for $p \in [n+1]$.
    Next, we add edges $\braces{u^{i}_{n+1}, v^{i}_1}$ and $\braces{v^{i}_{n+1}, u^{i}_1}$.
    Lastly, for $p \in [n]$, we add edges $\braces{v^{i}_p, w^{i,p}_1}$ and $\braces{w^{i,p}_{3}, v^{i}_{p+1}}$,
    as well as an edge from $u^{i}_p$ to one endpoint of $\hat{R}^{i}_p$, and an edge from $u^{i}_{p+1}$ to the other endpoint of $\hat{R}^{i}_p$.
    See \cref{fig:inapx_choice_gadget} for an illustration.
    Subsequently, add to $\M$ all pairs $(v^{i}_p, u^{i}_{p+1})$ and $(v^{i}_p, u^{i}_{p-1})$
    for $p \in [2, n]$,
    as well as the pairs $(v^{i}_1, u^{i}_2)$ and $(v^{i}_{n+1}, u^{i}_n)$.
    Let $\M_c \subseteq \M$ denote all such pairs added to $\M$ in this step of the construction.
    Intuitively, we will consider a one-to-one mapping between the vertex $v^i_p$ of $V_i$ being chosen and
    the vertices of $\hat{P}^{i}_p$ not being used to route any of the demands in $\M_c$.

    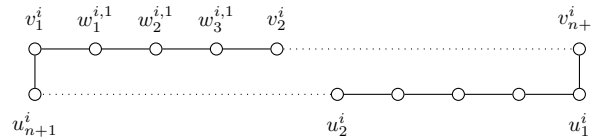
\begin{figure}[ht]
    \centering 
    \begin{tikzpicture}[scale=0.8, transform shape]
    
        \node[] () at (5,6.5) {$v^{i}_1$};
        \node[vertex] (v1) at (5,6) {};
        
        \node[] () at (6,6.5) {$w^{i,1}_1$};
        \node[vertex] (w11) at (6,6) {};
        \node[] () at (7,6.5) {$w^{i,1}_2$};
        \node[vertex] (w12) at (7,6) {};
        \node[] () at (8,6.5) {$w^{i,1}_{3}$};
        \node[vertex] (w1k) at (8,6) {};
    
        \node[] () at (9,6.5) {$v^{i}_2$};
        \node[vertex] (v2) at (9,6) {};
    
        \node[] () at (14,6.5) {$v^{i}_{n+1}$};
        \node[vertex] (vn) at (14,6) {};
    
        \draw[] (v1)--(w11)--(w12)--(w1k);
        \draw[] (v2)--(w1k);
        \draw[dotted] (v2)--(vn);
            
        \begin{scope}[shift={(0,-0.75)}]
            \node[] () at (14,5.5) {$u^{i}_1$};
            \node[vertex] (u1) at (14,6) {};
            
            \node[vertex] (r11) at (13,6) {};
            \node[vertex] (r12) at (12,6) {};
            \node[vertex] (r1k) at (11,6) {};
        
            \node[] () at (10,5.5) {$u^{i}_2$};
            \node[vertex] (u2) at (10,6) {};
        
            \node[] () at (5,5.5) {$u^{i}_{n+1}$};
            \node[vertex] (un) at (5,6) {};
        
            \draw[] (u1)--(r11)--(r12)--(r1k);
            \draw[] (u2)--(r1k);
            \draw[dotted] (u2)--(un);    
        \end{scope}
    
        \draw (v1)--(un);
        \draw (vn)--(u1);
    
        \end{tikzpicture}
        \caption{Choice gadget $\hat{C}_{i}$.}
        \label{fig:inapx_choice_gadget}
    \end{figure}

    \proofsubparagraph*{Adjacency vertices.}
    Let $E^{i,j} \neq \varnothing$.
    Introduce an \emph{adjacency vertex} $e_{i,j}$,
    and add edges $\braces{e_{i,j}, w^{i,p}_x}$ and
    $\braces{e_{i,j}, w^{j,p}_y}$ where $p \in [n]$, and $x, y \in [3]$ such that
    no other adjacency vertex is adjacent to them (since $s_i \leq 3$ for all $i \in [k]$,
    there always exist such $x$ and $y$).
    If $e = \braces{v^i_p, v^j_q} \in E^{i,j}$,
    then add the pair $(w^{i,p}_x, w^{j,q}_y)$ in $\M$.
    Notice that all adjacency vertices have disjoint neighborhoods,
    and let $F \subseteq V(H)$ be the set of all adjacency vertices, where $|F| = \sigma$.

    This concludes the construction of the instance $\mathcal{J}$.
    Notice that $H - F$ is a collection of $k$ choice gadgets, each of which is a cycle.
    Consequently, the graph obtained from $H-F$ by deleting a vertex per choice gadget is a collection of paths,
    thus both $\pw(H)$ and $\fvs(H)$ are at most $\bO(k)$.

    We first prove that if $\OPT(\mathcal{I}) = \sigma$, then $\OPT(\mathcal{J}) = \ell$.
    Consider a function $h \colon [k] \to [n]$ such that $G[\mathcal{V}]$ has $\sigma$ edges,
    where $\mathcal{V} = \setdef{v^i_{h(i)} \in V_i}{i \in [k]} \subseteq V(G)$.
    We construct a family of $\ell$ vertex-disjoint paths as follows.
    First, for every $i \in [k]$,
    we route two demands in~$\hat{C}_{i}$.
    In particular, we route the demands $(v^{i}_{h(i)}, u^{i}_{h(i) + 1})$ as well as
    $(v^{i}_{h(i)+1}, u^{i}_{h(i)})$,
    using the vertices of the shortest path of $\hat{C}_{i}$ in each case.        
    Note that in this step we have created $2k$ vertex-disjoint paths connecting terminal pairs belonging to $\M_c$,
    and in every gadget $\hat{C}_{i}$ the only unused vertices are those of $\hat{P}^{i}_{h(i)}$ and $\hat{R}^{i}_{h(i)}$.
    Then, consider the adjacency vertex~$e_{i,j}$, where $i,j \in [k]$.
    Route the demand $(w^{i,h(i)}_{x}, w^{j,h(j)}_{y})$ via~$e_{i,j}$,
    since both endpoints of the demand are its neighbors and
    have not been used in any path so far, for some $x,y \in [3]$.
    Such a demand indeed exists in~$\M$, since $\braces{v^i_{h(i)}, v^j_{h(j)}} \in E^{i,j}$;
    if that were not the case then~$G[\mathcal{V}]$ has less than~$\sigma$ edges.
    This procedure results in~$\sigma$ additional demands being routed.
    Notice that since the neighborhoods of all adjacency vertices are disjoint,
    the~$\ell$ resulting paths are indeed vertex-disjoint.

    It remains to prove that if $\OPT(\mathcal{I}) < c \cdot \sigma$,
    then $\OPT(\mathcal{J}) < c' \cdot \ell$ (or its contrapositive, as we will do in actuality).
    We start with \cref{claim:inapx_pw_helper}.

    \begin{claim}\label{claim:inapx_pw_helper}
        At most $2$ demands can be routed in a choice gadget $\hat{C}_{i}$ in $H-F$,
        in which case the only unused vertices are those of $\hat{P}^{i}_p$ and $\hat{R}^{i}_p$,
        for some $p \in [n]$.
        Additionally, it holds that $\OPT(\mathcal{J}) \leq \ell$.
    \end{claim}

    \begin{claimproof}
        Let $N = 3n + (n+1)$ and notice that $\hat{C}_{i}$ is a $C_{2N}$,
        thus there are exactly $2$ simple paths connecting any pair of its vertices.
        Moreover, the number of vertices in any simple path between $v^{i}_p$ and $u^{i}_p$,
        including said endpoints, is $N+1$.
        Consequently, to route any demand with both endpoints in $\hat{C}_{i}$,
        either $N+1 - 4$ or $N+1 + 4$ vertices are used.
        In case more than $2$ demands are routed,
        at least $3 (N+1 - 4) = 3N - 9$ vertices are used, which is a contradiction since
        $\hat{C}_{i}$ consists of $2N$ vertices.
        
        Assume that exactly $2$ demands are routed.
        Moreover, assume that a demand is routed using $N+1 + 4$ vertices.
        Then, both routed demands use at least $(N+1 + 4) + (N+1 - 4) > 2N$ vertices, which is a contradiction.
        Therefore, both demands that are routed in $\hat{C}_{i}$ use the shortest path connecting their endpoints.

        Let $(v^{i}_{p}, u^{i}_{q})$ and $(v^{i}_{p'}, u^{i}_{q'})$ denote the demands routed,
        where $p < p'$ and $p \in [n]$.
        It holds that $q = p+1$, since otherwise $q = p-1$,
        and $v^{i}_{p'}$ belongs to the shortest path connecting $(v^{i}_{p}, u^{i}_{q})$, a contradiction.
        Symmetrically, it follows that $q' = p'-1$.
        Lastly, it holds that $q > q'$, since otherwise $u^{i}_{q'}$ belongs to the shortest path connecting $(v^{i}_{p}, u^{i}_{q})$.
        Consequently, it follows that $q > q' \iff p+2 > p'$, which implies that $p < p' < p+2$, i.e., $p' = p+1$.
        In that case, the only unused vertices are those of $\hat{P}^{i}_p$ and $\hat{R}^{i}_p$.

        Notice that $H-F$ is a collection of $k$ choice gadgets.
        Since in each such gadget at most~$2$ demands can be routed,
        it follows that $\OPT(\mathcal{J}) \leq 2k + |F| = 2k + \sigma = \ell$.
    \end{claimproof}

    We now move on to prove that if $\OPT(\mathcal{J}) \geq c' \cdot \ell$, 
    then $\OPT(\mathcal{I}) \geq c \cdot \sigma$.
    Let $\mathcal{P}$ denote a collection of $\OPT(\mathcal{J})$ vertex-disjoint paths of $H$
    routing demands of $\M$, where $\OPT(\mathcal{J}) \geq c' \cdot \ell$.
    Additionally, let $\mathcal{C}^{\mathcal{P}}_j$ contain the choice
    gadgets $\hat{C}_i$ such that there exist
    \emph{exactly} $j$ paths in $\mathcal{P}$ that route demands in the graph induced by $\hat{C}_i$,
    for $j \in [0,2]$.
    Notice that $(\mathcal{C}^{\mathcal{P}}_0, \mathcal{C}^{\mathcal{P}}_1, \mathcal{C}^{\mathcal{P}}_2)$
    defines a partition of the choice gadgets due to \cref{claim:inapx_pw_helper}.
    Moreover, let $\mathcal{P}^F \subseteq \mathcal{P}$ be comprised of the paths of $\mathcal{P}$
    that contain vertices of $F$.
    We will say that a path $P \in \mathcal{P}^F$ \emph{intersects}~$\hat{C}_i$ if $P$ contains a vertex of $\hat{C}_i$.
    We define the \emph{loss} of a solution $\mathcal{P}$ to be $L_{\mathcal{P}} = \ell - |\mathcal{P}|$.
    Notice that due to \cref{claim:inapx_pw_helper} it holds that $L_{\mathcal{P}} = 2 |\mathcal{C}^{\mathcal{P}}_0| + |\mathcal{C}^{\mathcal{P}}_1| + (|F| - |\mathcal{P}^F|)$.
    
    It holds that $\OPT(\mathcal{J}) = |\mathcal{P}|
    = \ell - L_{\mathcal{P}} \geq c' \cdot \ell = \ell - (1 - c') \cdot \ell$,
    thus it follows that $L_{\mathcal{P}} \leq (1 - c') \cdot \ell$.
    Construct a new solution $\mathcal{P}'$ in the following way:
    for every $\hat{C}_i \in \mathcal{C}^{\mathcal{P}}_0 \cup \mathcal{C}^{\mathcal{P}}_1$,
    route two demands of $\M_c$ such that there exist two vertex-disjoint paths using only vertices of $\hat{C}_i$.
    Afterwards, remove any paths of $\mathcal{P}^F$ that are not vertex-disjoint with those.
    There are at most 3 vertices of $F$ adjacent to vertices of~$\hat{C}_i$,
    and by choosing the routed demands of~$\hat{C}_i$ in such a way that at least one
    of those vertices of $\hat{C}_i$ remains unused,
    it follows that at most~$2$ paths of $\mathcal{P}^F$ intersect $\hat{C}_i$.
    Thus it follows that $L_{\mathcal{P}'} \leq (|F| - |\mathcal{P}^F|) + 2 (|\mathcal{C}^\mathcal{P}_0| + |\mathcal{C}^\mathcal{P}_1|) \leq 2 L_{\mathcal{P}}$,        
    therefore $|\mathcal{P}'| = \ell - L_{\mathcal{P}'} \geq \ell - 2 L_{\mathcal{P}}$.
    Furthermore, notice that since for all $i \in [k]$ it holds that $\hat{C}_i \in \mathcal{C}^{\mathcal{P}'}_2$,
    due to \cref{claim:inapx_pw_helper} it follows that only the vertices of
    $\hat{P}^{i}_{h(i)}$ and $\hat{R}^i_{h(i)}$ remain unused by the paths of $\mathcal{P}'$ that
    route demands in $\M_c$,
    for some function $h \colon [k] \to [n]$.
    Consequently, for any routed demand $(w^{i,p}_x, w^{j,q}_y) \in \M \setminus \M_c$, it holds that $p = h(i)$ and $q = h(j)$.

    Let $\mathcal{V} = \setdef{v^i_{h(i)}}{i \in [k]}$, and notice that $|\mathcal{V} \cap V_i| = 1$
    for all $i \in [k]$.
    Let $A = |E(G[\mathcal{V}])|$ denote the number of edges present in the subgraph induced by $\mathcal{V}$.
    We will prove that $A \geq c \cdot \sigma$, in which case it follows that $\OPT(\mathcal{I}) \geq c \cdot \sigma$.
    Notice that $A \geq \ell - 2 L_{\mathcal{P}} - 2k = \sigma - 2 L_{\mathcal{P}}$,
    since this is the number of routed demands in $\M \setminus \M_c$ by $\mathcal{P}'$,
    while $(w^{i,h(i)}_x, w^{j,h(j)}_y) \in \M$ implies that $\braces{v^i_{h(i)}, v^j_{h(j)}} \in E^{i,j}$.

    It suffices to prove that $\sigma - 2 L_{\mathcal{P}} \geq c \cdot \sigma$.
    Since $\sigma - 2 L_{\mathcal{P}} \geq \sigma - 2 (1 - c') \cdot \ell$,
    we will prove $\sigma - 2 (1 - c') \cdot \ell \geq c \cdot \sigma$ instead,
    which is equivalent to $c' \geq 1 + \frac{\sigma}{2\ell} (c-1)$.
    Since $c-1 < 0$ and $\frac{\sigma}{2\ell} = \frac{1}{2} \cdot \frac{\sigma}{2k+\sigma}
    \geq \frac{1}{2} \cdot \frac{k / 2}{2k + k / 2} = \frac{1}{10}$,
    the statement holds for $c' = 1 + \frac{1}{10} (c-1) = \frac{9+c}{10}$.
\end{proof}

\section{XNLP-completeness}\label{sec:xnlp}

The \Wone-hardness results of~\cite{swat/EneMPR16,soda/MarxW15}
already imply that {\mNDP} is \Wone-hard parameterized by the pathwidth of the input graph.
Here we examine in more detail the parameterization solely by pathwidth,
and prove that in this case the problem is in fact XNLP-complete.
This complexity class was recently brought forth by Bodlaender, Groenland, Nederlof, and Swennenhuis~\cite{iandc/BodlaenderGNS24}
and consists of the parameterized problems such that an instance $(x,k)$,
where $x$ can be encoded with $n$ bits and $k$ denotes the parameter,
can be solved non-deterministically in time $f(k) n^{\bO(1)}$ and space $f(k) \log n$, for some computable function~$f$.

Such a completeness result in fact implies that {\mNDP} parameterized by pathwidth is W[$t$]-hard for all $t \in \N$.
To prove said result, we reduce from the XNLP-complete \CMC,
and use a construction quite similar to the one of \cref{thm:inapx_pw}.

\begin{theoremrep}[\appsymb]\label{thm:pw_xnlp}
    {\mNDP} parameterized by the pathwidth of the input graph is \XNLP-complete.
\end{theoremrep}

\begin{proof}

    We first argue that {\mNDP} parameterized by the pathwidth of the input graph belongs to \XNLP.
    Let $\mathcal{I} = (G, \M, \ell)$ be an instance of \mNDP,
    where $n = |\mathcal{I}|$ denotes the number of bits needed to encode $\mathcal{I}$,
    and assume that we are also given a path decomposition of $G$ of width $\pw$.
    We describe a non-deterministic algorithm for $\mathcal{I}$ that uses $\bO(\pw \log n)$ bits of memory.
    Fix an optimal solution, and observe that for each bag, at most $\pw$ paths of the solution intersect the bag.
    The algorithm guesses and stores for each bag both endpoints of the paths of the solution intersecting it
    ($2\pw \log n$) bits.
    Moreover, the algorithm guesses and stores for each vertex of the bag an integer from $[0, \ell]$,
    where $0$ indicates that this vertex is not used in the solution,
    otherwise said vertex is part of the $i$-th satisfied demand path.
    For vertices for which we have not stored $0$,
    we also remember their degree, i.e., how many of their neighbors that belong to the same path
    are present in the bags of the decomposition up to this point
    (this is an integer in $\{0,1,2\}$).
    In total, $\bO(\pw \log n)$ bits are used for this information.
    It is easy to guess and update these values when introducing a vertex.
    When Forgetting, the vertex must either have degree $2$, or alternatively degree $1$ and be an endpoint of its path.
    If both endpoints of a path have been Forgotten, we remove this from the list of active paths.
    Keep a counter of how many correct satisfied paths we have seen, using $\bO(\log n)$ additional bits,
    and check in the end that it is at least $\ell$.

    In order to prove the \XNLP-hardness,
    we present a \emph{parameterized logspace reduction}~\cite{iandc/BodlaenderGNS24} from \CMC,
    which is known to be XNLP-complete parameterized by $k$~\cite{iandc/BodlaenderGNS24},
    and is formally defined as follows:

    \problemdef{\CMC}
    {Graph $G=(V,E)$, a partition of $V$ into sets $V_1, \ldots, V_r$, as well as
    a coloring function $f \colon V \to [k]$,
    where for every $\braces{u,v} \in E$, if $u \in V_{i_1}$ and $v \in V_{i_2}$,
    then $|i_1 - i_2| \leq 1$.}
    {Determine whether there exists $W \subseteq V$ such that for every $i \in [r-1]$,
    $W \cap (V_i \cup V_{i+1})$ is a clique, and for all $i \in [r]$ and $j \in [k]$,
    there is a vertex $w \in W \cap V_i$ with $f(w) = j$.}

    On a high level, {\CMC} asks whether there exists a clique with $2k$ vertices in $V_i \cup V_{i+1}$ for each $i \in [r-1]$,
    containing a vertex of color $j$ both in $V_i$ and in $V_{i+1}$, for all colors $j \in [k]$.
    Importantly, the same vertices in $V_i$ are chosen in the clique for both $V_{i-1} \cup V_i$ and $V_i \cup V_{i+1}$.
    We call such a set a \emph{chained multicolored clique}.
    Furthermore, one can assume without loss of generality that $|\setdef{v \in V_i}{f(v) = j}| = n$, for every $i \in [r]$ and $j \in [k]$,
    since one can arbitrarily add a sufficient number of vertices of degree $0$.

    Let $(G,f,\setdef{V_i}{i \in [r]})$ be an instance of \CMC,
    and let $V_{i,j} = \braces{v^{i,j}_1, \ldots, v^{i,j}_n}$ denote the set of vertices belonging to $V_i$ that are of color $j$.
    We will construct in polynomial time an equivalent instance $(H, \M, \ell)$ of \mNDP,
    where $H$ is a graph of pathwidth $\pw(H) = \bO(k^2)$ and size $|V(H)| = \bO(|V(G)| \cdot k)$,
    $\M \subseteq \binom{V(H)}{2}$ is a set of demands,
    and $\ell = 2 r k + r \binom{k}{2} + (r-1) k^2$,
    such that $G$ has a chained multicolored clique if and only if at least $\ell$ demands of $\M$ can be routed.

    The construction will be very similar to the one of \cref{thm:inapx_pw}, albeit with a few differences.
    In particular, we will once again employ the use of the choice gadgets introduced there and we refer the reader to \cref{fig:inapx_choice_gadget}
    for an illustration.

    \proofsubparagraph*{Choice Gadget.}
    For every set $V_{i,j}$, where $i \in [r]$ and $j \in [k]$,
    we construct the \emph{choice gadget} $\hat{C}_{i,j}$ in the following way.
    First, for $p \in [n]$, we construct paths on vertex sets $\hat{P}^{i,j}_p = \setdef{w^{i,j,p}_q}{q \in [3k]}$,
    as well as paths $\hat{R}^{i,j}_p$ on $3k$ unnamed vertices each.
    Then, we introduce vertices $v^{i,j}_p$ and $u^{i,j}_p$, for $p \in [n+1]$.
    Next, we add edges $\braces{u^{i,j}_{n+1}, v^{i,j}_1}$ and $\braces{v^{i,j}_{n+1}, u^{i,j}_1}$.
    Lastly, for $p \in [n]$, we add edges $\braces{v^{i,j}_p, w^{i,j,p}_1}$ and $\braces{w^{i,j,p}_{3k}, v^{i,j}_{p+1}}$,
    as well as an edge from $u^{i,j}_p$ to one endpoint of $\hat{R}^{i,j}_p$, and an edge from $u^{i,j}_{p+1}$ to the other endpoint of $\hat{R}^{i,j}_p$.
    Subsequently, add to $\M$ all pairs $(v^{i,j}_p, u^{i,j}_{p+1})$ as well as $(v^{i,j}_p, u^{i,j}_{p-1})$, for $p \in [2, n]$,
    as well as the pairs $(v^{i,j}_1, u^{i,j}_2)$ and $(v^{i,j}_{n+1}, u^{i,j}_n)$.
    Let $\M_c \subseteq \M$ denote all such pairs added to $\M$ in this step of the construction.
    Intuitively, we will consider a one-to-one mapping between the vertex $v^{i,j}_p$ of $V_{i,j}$ belonging to a supposed chained multicolored clique of $G$ and
    the vertices of $\hat{P}^{i,j}_p$ not being used to route any of the demands in $\M_c$.

    \proofsubparagraph*{Adjacency vertices.}
    Next, we introduce some \emph{adjacency vertices} into the graph.
    For every $i \in [r]$ and $\braces{j_1,j_2} \in \binom{[k]}{2}$,
    introduce vertex $e^i_{j_1,j_2}$ and add edges $\braces{e^i_{j_1,j_2}, w^{i,j_1,p}_{j_2}}$ and
    $\braces{e^i_{j_1,j_2}, w^{i,j_2,p}_{j_1}}$ for all $p \in [n]$.
    Additionally, if $\braces{v^{i,j_1}_p, v^{i,j_2}_q} \in E(G)$,
    then add the pair $(w^{i,j_1,p}_{j_2}, w^{i,j_2,q}_{j_1})$ in $\M$.
    Next, for every $i \in [r-1]$ and $j,j' \in [k]$,
    introduce vertex $e^{i,i+1}_{j,j'}$ and add edges $\braces{e^{i,i+1}_{j,j'}, w^{i,j,p}_{2k + j'}}$ and
    $\braces{e^{i,i+1}_{j,j'}, w^{i+1,j',p}_{k + j}}$ for all $p \in [n]$.
    Additionally, if $\braces{v^{i,j}_p, v^{i+1,j'}_q} \in E(G)$,
    then add the pair $(w^{i,j,p}_{2k + j'}, w^{i+1,j',q}_{k + j})$ in $\M$.
    Notice that all adjacency vertices have disjoint neighborhoods.

    A \emph{set gadget} $\hat{S}_i$, where $i \in [r]$,
    is composed of the choice gadgets $\hat{C}_{i,j}$ for all $j \in [k]$ as well as all the adjacency vertices $e^i_{j,j'}$, where $\braces{j,j'} \in \binom{[k]}{2}$.

    This concludes the construction of $H$.
    For an illustration, see \cref{fig:hardness_xnlp}.
    It holds that $|V(H)| = rk(6nk + 2n + 2) + r \binom{k}{2} + (r - 1) k^2 = \bO(r n k^2)$,
    while $|V(G)| = r n k$, therefore $|V(H)| = \bO(|V(G)| \cdot k)$ follows.
    It remains to argue about the space usage of the algorithm that constructs the instance $(H, \M, \ell)$.
    Notice that in order to uniquely identify the vertices of $H$, it suffices to keep track of $5$ different indices:
    one identifies the kind of vertex ($u$, $v$, $w$, $e^i_{j_1, j_2}$, $e^{i,i+1}_{j,j'}$,
    or a vertex of some path $\hat{R}$ in a choice gadget),
    and the rest are used to distinguish among vertices of the same kind.
    In that case, the number of all such encodings is $(n r k)^{\bO(1)}$.
    Moreover, given one such encoding, we can determine the other encodings (i.e., vertices)
    with which the first either has an edge, or takes part in a demand,
    since it suffices to only alter some of these indices by some constant,
    thus the space usage is $\bO(\log (n k r))$.

    \begin{figure}[ht]
    \centering 
    \begin{tikzpicture}[scale=1, transform shape]
    
        \node[rectangle,draw,minimum width=1.5cm,minimum height = 1.1cm] (ci1) at (5,5) {$\hat{C}_{i,1}$};
    
        \node[] () at (5,7.6) {$\vdots$};
    
        \node[rectangle,draw,minimum width=1.5cm,minimum height = 1.1cm] (cik) at (5,10) {$\hat{C}_{i,k}$};

        \node[] () at (4.6,7.5) {$e^i_{1,k}$};
        \node[vertex] (ei1k) at (4.2,7.5) {};
        \draw[dashed] (ei1k)--(ci1);
        \draw[dashed] (ei1k)--(cik);
    
        \begin{scope}[shift={(-5,0)}]
            \node[rectangle,draw,minimum width=1.5cm,minimum height = 1.1cm] (ci-1) at (5,5) {$\hat{C}_{i-1,1}$};
        
            \node[] () at (5,7.6) {$\vdots$};
        
            \node[rectangle,draw,minimum width=1.5cm,minimum height = 1.1cm] (ci-k) at (5,10) {$\hat{C}_{i-1,k}$};

            \node[] () at (3.8,7.5) {$e^{i-1}_{1,k}$};
            \node[vertex] (ei-1k) at (4.2,7.5) {};
            \draw[dashed] (ei-1k)--(ci-1);
            \draw[dashed] (ei-1k)--(ci-k);
        \end{scope}
    
        \node[] () at (2.5,4.5) {$e^{i-1,i}_{1,1}$};
        \node[vertex] (ei-i11) at (2.5,5) {};
        \draw[dashed] (ei-i11)--(ci1);
        \draw[dashed] (ei-i11)--(ci-1);
    
        \node[] () at (2.5,10.5) {$e^{i-1,i}_{k,k}$};
        \node[vertex] (ei-ikk) at (2.5,10) {};
        \draw[dashed] (ei-ikk)--(cik);
        \draw[dashed] (ei-ikk)--(ci-k);
    
        \node[] () at (2.5,6.5) {$e^{i-1,i}_{1,k}$};
        \node[vertex] (ei-i1k) at (2.5,7) {};
        \draw[dashed] (ei-i1k)--(ci-1);
        \draw[dashed] (ei-i1k)--(cik);
    
        \node[] () at (2.5,8.5) {$e^{i-1,i}_{k,1}$};
        \node[vertex] (ei-ik1) at (2.5,8) {};
        \draw[dashed] (ei-ik1)--(ci-k);
        \draw[dashed] (ei-ik1)--(ci1);
    
        \begin{scope}[shift={(5,0)}]
            \node[rectangle,draw,minimum width=1.5cm,minimum height = 1.1cm] (ci+1) at (5,5) {$\hat{C}_{i+1,1}$};
        
            \node[] () at (5,7.6) {$\vdots$};
        
            \node[rectangle,draw,minimum width=1.5cm,minimum height = 1.1cm] (ci+k) at (5,10) {$\hat{C}_{i+1,k}$};

            \node[] () at (6.4,7.5) {$e^{i+1}_{1,k}$};
            \node[vertex] (ei+1k) at (5.8,7.5) {};
            \draw[dashed] (ei+1k)--(ci+1);
            \draw[dashed] (ei+1k)--(ci+k);

            \node[] () at (2.5,4.5) {$e^{i,i+1}_{1,1}$};
            \node[vertex] (eii+11) at (2.5,5) {};
            \draw[dashed] (eii+11)--(ci1);
            \draw[dashed] (eii+11)--(ci+1);
        
            \node[] () at (2.5,10.5) {$e^{i,i+1}_{k,k}$};
            \node[vertex] (eii+kk) at (2.5,10) {};
            \draw[dashed] (eii+kk)--(cik);
            \draw[dashed] (eii+kk)--(ci+k);
        
            \node[] () at (2.5,6.5) {$e^{i,i+1}_{1,k}$};
            \node[vertex] (eii+1k) at (2.5,7) {};
            \draw[dashed] (eii+1k)--(ci1);
            \draw[dashed] (eii+1k)--(ci+k);
        
            \node[] () at (2.5,8.5) {$e^{i,i+1}_{k,1}$};
            \node[vertex] (eii+k1) at (2.5,8) {};
            \draw[dashed] (eii+k1)--(ci+1);
            \draw[dashed] (eii+k1)--(cik);
        
        \end{scope}

        \end{tikzpicture}
        \caption{Part of graph $H$ regarding set gadgets $\hat{S}_{i-1}$, $\hat{S}_{i}$, and $\hat{S}_{i+1}$.
        The vertices denote the middle vertex of the corresponding adjacency path,
        while dashed lines indicate that the respective adjacency path has neighbors in the corresponding choice gadget.
        }
        \label{fig:hardness_xnlp}
    \end{figure}

    \begin{lemma}\label{lem:xnlp_pw}
        It holds that $\pw(H) = \bO(k^2)$.
    \end{lemma}

    \begin{nestedproof}
        In the following, let $j_1,j_2,j$, and $j'$ such that $\braces{j_1,j_2} \in \binom{[k]}{2}$ and $j, j' \in [k]$.
        Start with $i=1$ and consider a path decomposition that consists of bags containing all vertices
        $e^{i-1,i}_{j,j'}$, $e^i_{j_1,j_2}$, and $e^{i,i+1}_{j,j'}$,
        and then going through all choice gadgets of level $i$ one by one,
        before moving on to level $i+1$ and so on, until level $r$.
        %
        %
        %
\end{nestedproof}

    \begin{lemma}\label{lem:xnlp_cor1}
        If $G$ contains a chained multicolored clique,
        then $(H, \M, \ell)$ is a Yes instance of \mNDP.
    \end{lemma}

    \begin{nestedproof}
        Let $\mathcal{W} \subseteq V(G)$ be a chained multicolored clique of $G$,
        consisting of vertices $\setdef{v^{i,j}_{s(i,j)}}{i \in [r], j \in [k]}$.
        We construct a family of $\ell$ vertex-disjoint paths as follows.

        First, for every $i \in [r]$ and $j \in [k]$,
        we route two demands in $\hat{C}_{i,j}$.
        In particular, we route the demands $(v^{i,j}_{s(i,j)}, u^{i,j}_{s(i,j) + 1})$ as well as $(v^{i,j}_{s(i,j)+1}, u^{i,j}_{s(i,j)})$,
        using the vertices of the shortest path of $\hat{C}_{i,j}$ in each case.
        Note that in this step we have created $2 r k$ vertex-disjoint paths connecting terminal pairs belonging to $\M_c$,
        and in every gadget $\hat{C}_{i,j}$ the only unused vertices are those of $\hat{P}^{i,j}_{s(i,j)}$ and $\hat{R}^{i,j}_{s(i,j)}$.

        Then, for $i \in [r]$ and $\braces{j_1, j_2} \in \binom{[k]}{2}$, consider the adjacency vertex $e^i_{j_1,j_2}$.
        Route the demand $(w^{i,j_1,s(i,j_1)}_{j_2}, w^{i,j_2,s(i,j_2)}_{j_1})$ via said adjacency vertex,
        since both endpoints of the demand are its neighbors and have not been used in any path so far.
        This procedure results in $r \binom{k}{2}$ additional demands being routed.

        Lastly, for $i \in [r-1]$ and $j,j' \in [k]$, consider the adjacency vertex $e^{i, i+1}_{j,j'}$.
        Route the demand $(w^{i,j,s(i,j)}_{2k + j'}, w^{i+1,j',s(i+1, j')}_{k + j})$ via said adjacency vertex,
        since both endpoints of the demand are its neighbors and have not been used in any path so far.
        Such a demand indeed exists in $\M$, since $\braces{v^{i,j}_{s(i,j)}, v^{i+1,j'}_{s(i+1,j')}} \in E(G)$.
        This procedure results in $(r-1) k^2$ additional demands being routed.

        Notice that since the neighborhoods of all adjacency vertices are disjoint,
        the $\ell$ resulting paths are indeed vertex disjoint.
    \end{nestedproof}

    \begin{lemma}\label{lem:xnlp_cor2}
        If $(H, \M, \ell)$ is a Yes instance of \mNDP,
        then $G$ contains a chained multicolored clique.
    \end{lemma}
    
    \begin{nestedproof}
        Let $\mathcal{P} = \braces{P_1, \ldots, P_\ell}$ be a set of $\ell$ vertex-disjoint paths connecting terminal pairs of $\M$ in $H$.
        Assume without loss of generality that said paths are simple,
        as well as that the only edges among vertices of a path appear between its consecutive vertices.
        Moreover, define $\mathcal{P}_{c} = \setdef{P_i \in \mathcal{P}}{P_i \text{ routes a demand in } \M_{c}}$.
        Set $S$ to be the set consisting of all the adjacency vertices in $H$, where $|S| = r \binom{k}{2} + (r-1) k^2$.

        Notice that at most $|S|$ paths of $\mathcal{P}$ contain vertices of $S$, consequently at least $\ell - |S| = 2 r k$ paths
        route demands using only vertices of $H-S$.
        Moreover, $H-S$ is a collection of $rk$ disconnected choice gadgets.
        We start with \cref{claim:xnlp_demands_per_choice_gadget},
        whose proof is omitted since it is analogous to the one of \cref{claim:inapx_pw_helper}.

        \begin{claim}\label{claim:xnlp_demands_per_choice_gadget}
            At most $2$ demands can be routed in a choice gadget $\hat{C}_{i,j}$ in $H-S$,
            in which case the only unused vertices are those of $\hat{P}^{i,j}_p$ and $\hat{R}^{i,j}_p$, for some $p \in [n]$.
        \end{claim}

        Consequently, due to \cref{claim:xnlp_demands_per_choice_gadget},
        it follows that exactly $|S|$ paths of $\mathcal{P}$ contain vertices of $S$.
        Moreover, in each choice gadget $\hat{C}_{i,j}$,
        exactly $2$ demands are routed, for a total of $2rk$ demands of $\M_{c}$,
        leaving unused only the vertices of $\hat{P}^{i,j}_p$ and $\hat{R}^{i,j}_p$ for some $p \in [n]$.
        Next we prove that any path not routing a demand of $\M_c$ contains exactly $3$ vertices.

        \begin{claim}\label{claim:xnlp_size_of_path}
            If $P \in \mathcal{P} \setminus \mathcal{P}_c$, then $P$ is of size $3$.
        \end{claim}

        \begin{claimproof}
            Since $P \notin \mathcal{P}_c$, it contains a single vertex of $S$, denoted by $s$.
            Let $\hat{C}_{i_1,j_1}$ and $\hat{C}_{i_2,j_2}$ denote the $2$ choice gadgets $s$ has neighbors in.
            Set $S' = S \setminus \braces{s}$.
            Then, it holds that $P$ routes one demand of $\M \setminus \M_{c}$, using the vertices of the graph $H - S'$.
            In that case, one endpoint of such a demand can only be a vertex of $\hat{C}_{i_1,j_1}$ and the other of $\hat{C}_{i_2,j_2}$,
            and due to the construction of $(H,\M,\ell)$, it holds that for every such demand, both its endpoints are neighbors of $s$.
        \end{claimproof}


        Let $\mathcal{W} \subseteq V(G)$ be a set of cardinality $rk$, containing vertex $v^{i,j}_{s(i,j)} \in V_i$ if,
        for choice gadget $\hat{C}_{i,j}$,
        it holds that the vertices of $\hat{P}^{i,j}_{s(i,j)}$ are not used to route demands of $\M_{c}$.
        Notice that $|\setdef{w \in W \cap V_i}{f(w) = j}| = 1$, for all $i \in [r]$ and $j \in [k]$.
        We will prove that $\mathcal{W}$ is a chained multicolored clique of $G$.

        Let $i_1, i_2, j_1$ and $j_2$ such that (i) either $i_1 = i_2$ and $\braces{j_1,j_2} \in \binom{[k]}{2}$,
        (ii) or $i_2 = i_1 + 1$ and $j_1, j_2 \in [k]$.
        Let $v^{i_1,j_1}_{s(i_1,j_1)}, v^{i_2,j_2}_{s(i_2,j_2)}$ belong to $\mathcal{W}$.
        Let $P \in \mathcal{P}$ denote the path containing vertex $e$,
        where $e = e^{i_1}_{j_1,j_2}$ in case (i) and $e = e^{i_1,i_2}_{j_1,j_2}$ otherwise.
        Due to \cref{claim:xnlp_size_of_path}, it holds that $P$ is comprised of $e$,
        as well as two of its neighbors, one in $\hat{C}_{i_1,j_1}$ and one in $\hat{C}_{i_2,j_2}$.
        Since the only neighbors of $e$ that are not used by paths in $\mathcal{P}_{c}$ are
        $w^{i_1,j_1,s(i_1,j_1)}_{c_1 + j_2}$ and $w^{i_2,j_2,s(i_2,j_2)}_{c_2 + j_1}$,
        we infer that $(w^{i_1,j_1,s(i_1,j_1)}_{c_1 + j_2}, w^{i_2,j_2,s(i_2,j_2)}_{c_2 + j_1}) \in \M$,
        where $c_1 = c_2 = 0$ in case (i) and $c_1 = 2k$, $c_2 = k$ otherwise.
        Consequently, $\braces{v^{i_1,j_1}_{s(i_1,j_1)}, v^{i_2,j_2}_{s(i_2,j_2)}} \in E(G)$.
        Since this holds for any two such vertices belonging to $\mathcal{W}$, it follows that $G$ has a chained multicolored clique.
    \end{nestedproof}

    Therefore, in polynomial time and with logarithmic space,
    we can construct a graph $H$,
    where $\pw = \bO(k^2)$ due to \cref{lem:xnlp_pw},
    as well as a set of pairs $\M$,
    such that, due to \cref{lem:xnlp_cor1,lem:xnlp_cor2},
    deciding whether at least $\ell$ pairs of $\M$ can be routed 
    is equivalent to deciding whether $G$ has a chained multicolored clique.
\end{proof}



\section{Refining Hardness for Bounded Tree-depth Graphs}\label{sec:td_lb}

In this section, we refine the hardness result of
Ene, Mnich, Pilipczuk, and Risteski~\cite{swat/EneMPR16} for bounded tree-depth graphs,
by employing a recursive structure introduced in~\cite{toct/LampisV24}.
The reduction of~\cite{swat/EneMPR16} starts from an instance $(G,k)$ of \kMC,
and produces an equivalent instance of {\mNDP} on a graph of tree-depth, vertex integrity, and feedback vertex number $\bO(k^2)$, implying a $n^{o(\sqrt{\td})}$ lower bound under the ETH.
We refine their approach, resulting in a reduction that keeps tree-depth linear in $k$,
thereby improving the lower bound to $n^{o(\td)}$ under the ETH.
As a consequence of our result, it follows that the standard $n^{\bO(\tw)}$ algorithm for the problem is optimal,
even if one considers the class of graphs of bounded tree-depth.

In order to achieve this result, we combine ideas from both~\cite{swat/EneMPR16} and~\cite{toct/LampisV24}.
On a high level, the reduction of~\cite{swat/EneMPR16} consists of $k$ choice gadgets,
each of which is used to encode the vertex which is chosen to take part in a supposed clique per color class.
Afterwards, it suffices to add $\binom{k}{2}$ vertices in order to verify the existence of edges among
all the chosen vertices of the color classes.
The deletion of those $\binom{k}{2}$ vertices then gives the bounds for the tree-depth of the graph.
In order to avoid this quadratic dependence,
we make use of a recursive structure introduced in~\cite{toct/LampisV24}
meant to verify the existence of edges between
the chosen vertices, while keeping the tree-depth of the resulting graph linear in $k$.

\begin{theorem}\label{thm:td_lb}
    For any computable function $f$, if there exists an algorithm that solves
    {\mNDP} in time $f(\td) n^{o(\td)}$, where $\td$ denotes the tree-depth
    of the input graph, then the ETH is false.
\end{theorem}

\begin{proof}
    Let $(G,k)$ be an instance of \kMC,
    such that every vertex of $G$ has a self loop,
    i.e., $\braces{v,v} \in E(G)$, for all $v \in V(G)$.
    Recall that we assume that~$G$ is given to us partitioned into $k$ independent sets $V_1, \ldots, V_k$,
    where $V_i = \braces{v^i_1, \ldots, v^i_n}$.
    Assume without loss of generality that $k = 2^z$, for some $z \in \N$
    (one can do so by adding dummy independent sets connected to all the other vertices of the graph).
    Moreover, let $E^{i_1,i_2} \subseteq E(G)$ denote the edges of $G$ with one endpoint in $V_{i_1}$ and the other in $V_{i_2}$.
    We will construct in polynomial time an equivalent instance $(H, \M, \ell)$ of \mNDP,
    where $H$ is a graph of tree-depth $\td(H) = \bO(k)$,
    feedback vertex number $\fvs(H) = \bO(k^2)$,
    and size $|V(H)| = n^{\bO(1)}$,
    $\M \subseteq \binom{V(H)}{2}$ is a set of demands,
    and $\ell$ is an integer,
    such that $G$ has a $k$-clique if and only if at least $\ell$ demands of $\M$ can be routed in $H$.

    \proofsubparagraph*{Choice Gadget.}
    For an independent set $V_i$, we construct the \emph{choice gadget} $\hat{C}_i$ as depicted in \cref{fig:td_lb_choice_gadget}.
    We first construct paths on vertex sets $\hat{P}^i_j = \braces{v^{i,j}_1, v^{i,j}_2, v^{i,j}_3}$, where $j \in [n]$.
    Afterwards, for every $j \in [2,n]$,
    we introduce vertices $\alpha^i_j$ and $\beta^i_j$, connecting them with $v^{i,1}_1, v^{i,j}_1$ and $v^{i,1}_3, v^{i,j}_3$
    respectively.
    Moreover, we add the pair $(\alpha^i_j, \beta^i_j)$ to $\M$.
    Intuitively, we will consider a one-to-one mapping between the vertex $v^i_j$ of $V_i$ belonging to a supposed $k$-clique of $G$ and
    the vertices of $\hat{P}^i_j$ not being used to route any of the demands added in this step.
    
    \proofsubparagraph*{Copy Gadget.}
    Given two instances $\mathcal{I}_1$, $\mathcal{I}_2$ of a choice gadget $\hat{C}_i$,
    when we say that we add a \emph{copy gadget $(\mathcal{I}_1, \mathcal{I}_2, t)$}, where $t \in \braces{1,2}$,
    we introduce a \emph{copy vertex} $g$, connect it with all vertices $v^{i,j}_{t+1}$ of $\mathcal{I}_1$ and all vertices $v^{i,j}_1$ of $\mathcal{I}_2$ for $j \in [n]$
    and add all the pairs $(v^{i,j}_{t+1}, v^{i,j}_1)$ to~$\M$.

    \begin{figure}[ht]
        \centering 
          \begin{subfigure}[b]{0.4\textwidth}
          \centering
            \begin{tikzpicture}[scale=0.7, transform shape]
                \node[vertex] (r1) at (5,6) {};
                \node[] () at (4.5,6) {$v^{i,1}_1$};
                \node[vertex] (t11) at (5,4.5) {};
                \node[] () at (4.5,4.5) {$v^{i,1}_2$};
                \node[vertex] (t12) at (5,3) {};
                \node[] () at (4.5,3) {$v^{i,1}_3$};
                \draw[] (r1)--(t11)--(t12);
                
                \begin{scope}[shift={(1,0)}]
                    \node[vertex] (r2) at (5,6) {};
                    \node[vertex] (t21) at (5,4.5) {};
                    \node[vertex] (t22) at (5,3) {};
                    \draw[] (r2)--(t21)--(t22);
                \end{scope}
                
                \begin{scope}[shift={(0.5,0)}]
                    \node[vertex] (a1) at (5,6.5) {};
                    \node[vertex] (b1) at (5,2.5) {};
                    \draw[] (a1) edge [bend right] (r1);
                    \draw[] (a1) edge [bend left] (r2);
                    \draw[] (b1) edge [bend right] (t22);
                    \draw[] (b1) edge [bend left] (t12);
                \end{scope}
                
                \node[] () at (6.5,4.5) {$\cdots$};
                
                \begin{scope}[shift={(2,0)}]
                    \node[vertex] (rn) at (5,6) {};
                    \node[] () at (5.5,6) {$v^{i,n}_1$};
                    \node[vertex] (t1n) at (5,4.5) {};
                    \node[] () at (5.5,4.5) {$v^{i,n}_2$};
                    \node[vertex] (t2n) at (5,3) {};
                    \node[] () at (5.5,3) {$v^{i,n}_3$};
                    \draw[] (rn)--(t1n)--(t2n);
                \end{scope}
                
                \begin{scope}[shift={(1,0)}]
                    \node[vertex] (an) at (5,7) {};
                    \node[] () at (5,7.4) {$\alpha^i_n$};
                    \node[vertex] (bn) at (5,2) {};
                    \node[] () at (5,1.7) {$\beta^i_n$};
                    \draw[] (an) edge [bend right] (r1);
                    \draw[] (an) edge [bend left] (rn);
                    \draw[] (bn) edge [bend right] (t2n);
                    \draw[] (bn) edge [bend left] (t12);
                \end{scope}
            \end{tikzpicture}
            \caption{Choice gadget $\hat{C}_i$.}
            \label{fig:td_lb_choice_gadget}
          \end{subfigure}
        \begin{subfigure}[b]{0.4\linewidth}
        \centering
            \begin{tikzpicture}[scale=0.65, transform shape]
        
                \node[vertex] (r1) at (5,6) {};
                \node[] () at (4.5,6) {$v^{i,1}_1$};
                \node[vertex] (t11) at (5,4.5) {};
                \node[] () at (4.5,4.5) {$v^{i,1}_2$};
                \node[vertex] (t12) at (5,3) {};
                \node[] () at (4.5,3) {$v^{i,1}_3$};
                \draw[] (r1)--(t11)--(t12);
                
                \begin{scope}[shift={(1,0)}]
                    \node[vertex] (r2) at (5,6) {};
                    \node[vertex] (t21) at (5,4.5) {};
                    \node[vertex] (t22) at (5,3) {};
                    \draw[] (r2)--(t21)--(t22);
                \end{scope}
                
                \begin{scope}[shift={(0.5,0)}]
                    \node[vertex] (a1) at (5,6.5) {};
                    \node[vertex] (b1) at (5,2.5) {};
                    \draw[] (a1) edge [bend right] (r1);
                    \draw[] (a1) edge [bend left] (r2);
                    \draw[] (b1) edge [bend right] (t22);
                    \draw[] (b1) edge [bend left] (t12);
                \end{scope}
                
                \node[] () at (6.5,4.5) {$\cdots$};
                
                \begin{scope}[shift={(2,0)}]
                    \node[vertex] (rn) at (5,6) {};
                    \node[] () at (5.5,6) {$v^{i,n}_1$};
                    \node[vertex] (t1n) at (5,4.5) {};
                    \node[] () at (5.5,4.2) {$v^{i,n}_2$};
                    \node[vertex] (t2n) at (5,3) {};
                    \node[] () at (5.5,3) {$v^{i,n}_3$};
                    \draw[] (rn)--(t1n)--(t2n);
                \end{scope}
                
                \begin{scope}[shift={(1,0)}]
                    \node[vertex] (an) at (5,7) {};
                    \node[] () at (5,7.4) {$\alpha^i_n$};
                    \node[] () at (5,8.5) {$\mathcal{I}_1$};
                    \node[vertex] (bn) at (5,2) {};
                    \node[] () at (5,1.7) {$\beta^i_n$};
                    \draw[] (an) edge [bend right] (r1);
                    \draw[] (an) edge [bend left] (rn);
                    \draw[] (bn) edge [bend right] (t2n);
                    \draw[] (bn) edge [bend left] (t12);
                \end{scope}
        
                \begin{scope}[shift={(6,0)}]
                    \node[vertex] (Ar1) at (5,6) {};
                    \node[] () at (4.5,6) {$v^{i,1}_1$};
                    \node[vertex] (At11) at (5,4.5) {};
                    \node[] () at (4.5,4.2) {$v^{i,1}_2$};
                    \node[vertex] (At12) at (5,3) {};
                    \node[] () at (4.5,3) {$v^{i,1}_3$};
                    \draw[] (Ar1)--(At11)--(At12);
                    
                    \begin{scope}[shift={(1,0)}]
                        \node[vertex] (Ar2) at (5,6) {};
                        \node[vertex] (At21) at (5,4.5) {};
                        \node[vertex] (At22) at (5,3) {};
                        \draw[] (Ar2)--(At21)--(At22);
                    \end{scope}
                    
                    \begin{scope}[shift={(0.5,0)}]
                        \node[vertex] (Aa1) at (5,6.5) {};
                        \node[vertex] (Ab1) at (5,2.5) {};
                        \draw[] (Aa1) edge [bend right] (Ar1);
                        \draw[] (Aa1) edge [bend left] (Ar2);
                        \draw[] (Ab1) edge [bend right] (At22);
                        \draw[] (Ab1) edge [bend left] (At12);
                    \end{scope}
                    
                    \node[] () at (6.5,4.5) {$\cdots$};
                    
                    \begin{scope}[shift={(2,0)}]
                        \node[vertex] (Arn) at (5,6) {};
                        \node[] () at (5.5,6) {$v^{i,n}_1$};
                        \node[vertex] (At1n) at (5,4.5) {};
                        \node[] () at (5.5,4.5) {$v^{i,n}_2$};
                        \node[vertex] (At2n) at (5,3) {};
                        \node[] () at (5.5,3) {$v^{i,n}_3$};
                        \draw[] (Arn)--(At1n)--(At2n);
                    \end{scope}
                    
                    \begin{scope}[shift={(1,0)}]
                        \node[vertex] (Aan) at (5,7) {};
                        \node[] () at (5,8.5) {$\mathcal{I}_2$};
                        \node[] () at (5,7.4) {$\alpha^i_n$};
                        \node[vertex] (Abn) at (5,2) {};
                        \node[] () at (5,1.7) {$\beta^i_n$};
                        \draw[] (Aan) edge [bend right] (Ar1);
                        \draw[] (Aan) edge [bend left] (Arn);
                        \draw[] (Abn) edge [bend right] (At2n);
                        \draw[] (Abn) edge [bend left] (At12);
                    \end{scope}
                \end{scope}

                \node[vertex] (g) at (9,4.6) {};
                \node[] () at (9,5) {$g$};
                \draw[] (t11) edge [bend left] (g);
                \draw[] (t21) edge [bend left] (g);
                \draw[] (t1n) edge [bend left] (g);
                \draw[] (Ar1) edge [bend left] (g);
                \draw[] (Ar2) edge [bend left] (g);
                \draw[] (Arn) edge [bend left] (g);

            \end{tikzpicture}
            \caption{Addition of a copy gadget $(\mathcal{I}_1, \mathcal{I}_2, 1)$.}
            \label{fig:td_lb_copy_choice_gadget}
          \end{subfigure}
        \caption{Choice gadgets and how to copy them.}
    \end{figure}

    \proofsubparagraph*{Adjacency Gadget.}
    Let $i_1, i_2, i'_1, i'_2 \in [k]$.
    For $i_1 \leq i_2$ and $i'_1 \leq i'_2$, we define the \emph{adjacency gadget}%
    \footnote{For some high-level intuition regarding the adjacency gadgets,
    we refer the reader to~\cite[Section~4]{toct/LampisV24}.}
    $\hat{A}(i_1, i_2, i'_1, i'_2)$ as follows:
    \begin{itemize}
        \item Consider first the case when $i_1 = i_2 = i$ and $i'_1 = i'_2 = i'$.
        Let the adjacency gadget contain instances of the choice gadgets $\hat{C}_{i}$ and $\hat{C}_{i'}$,
        as well as a \emph{validation vertex} $m_{i,i'}$.
        Add edges between $m_{i,i'}$ and all vertices $v^{i,j}_2$ and $v^{i',j}_2$ for $j \in [n]$.
        If $e = \braces{v^{i}_{j}, v^{i'}_{j'}} \in E^{i,i'}$, then add the pair $(v^{i,j}_2, v^{i',j'}_2)$ to $\M$.
    
        \item Now consider the case when $i_1 < i_2$ and $i'_1 < i'_2$.
        Then, let $\hat{A}(i_1,i_2,i'_1,i'_2)$ contain instances of choice gadgets $\hat{C}_{i}$ and $\hat{C}_{i'}$, where $i \in [i_1,i_2]$ and $i' \in [i'_1, i'_2]$,
        which we will refer to as the \emph{original choice gadgets} of $\hat{A}(i_1,i_2,i'_1,i'_2)$,
        as well as the adjacency gadgets
        \begin{multicols}{2}
        \begin{itemize}
            \item $\hat{A}\parens*{i_1, \floor*{\frac{i_1+i_2}{2}}, i'_1, \floor*{\frac{i'_1+i'_2}{2}}}$,
            \item $\hat{A}\parens*{i_1, \floor*{\frac{i_1+i_2}{2}}, \ceil*{\frac{i'_1+i'_2}{2}}, i'_2}$,
            \item $\hat{A}\parens*{\ceil*{\frac{i_1+i_2}{2}}, i_2, i'_1, \floor*{\frac{i'_1+i'_2}{2}}}$,
            \item $\hat{A}\parens*{\ceil*{\frac{i_1+i_2}{2}}, i_2, \ceil*{\frac{i'_1+i'_2}{2}}, i'_2}$.
        \end{itemize}
        \end{multicols}
        Lastly, let $\mathcal{I}$ denote the original choice gadget $\hat{C}_p$, where $p \in [i_1, i_2] \cup [i'_1, i'_2]$.
        Notice that there are exactly two instances of choice gadget $\hat{C}_p$ appearing as original choice gadgets in the adjacency gadgets just introduced,
        say instances $\mathcal{I}_1$ and $\mathcal{I}_2$.
        Add copy gadgets $(\mathcal{I}, \mathcal{I}_1, 1)$ and $(\mathcal{I}, \mathcal{I}_2, 2)$.
    \end{itemize}
        
    \begin{figure}[ht]
    \centering
        \begin{tikzpicture}[scale=0.7, transform shape]
    
            \node[vertex] (r1) at (5,6) {};
            \node[] () at (4.5,6) {$v^{i,1}_1$};
            \node[vertex] (t11) at (5,4.5) {};
            \node[] () at (4.5,4.5) {$v^{i,1}_2$};
            \node[vertex] (t12) at (5,3) {};
            \node[] () at (4.5,3) {$v^{i,1}_3$};
            \draw[] (r1)--(t11)--(t12);
            
            \begin{scope}[shift={(1,0)}]
                \node[vertex] (r2) at (5,6) {};
                \node[vertex] (t21) at (5,4.5) {};
                \node[vertex] (t22) at (5,3) {};
                \draw[] (r2)--(t21)--(t22);
            \end{scope}
            
            \begin{scope}[shift={(0.5,0)}]
                \node[vertex] (a1) at (5,6.5) {};
                \node[vertex] (b1) at (5,2.5) {};
                \draw[] (a1) edge [bend right] (r1);
                \draw[] (a1) edge [bend left] (r2);
                \draw[] (b1) edge [bend right] (t22);
                \draw[] (b1) edge [bend left] (t12);
            \end{scope}
            
            \node[] () at (6.5,4.5) {$\cdots$};
            
            \begin{scope}[shift={(2,0)}]
                \node[vertex] (rn) at (5,6) {};
                \node[] () at (5.5,6) {$v^{i,n}_1$};
                \node[vertex] (t1n) at (5,4.5) {};
                \node[] () at (5.5,4.3) {$v^{i,n}_2$};
                \node[vertex] (t2n) at (5,3) {};
                \node[] () at (5.5,3) {$v^{i,n}_3$};
                \draw[] (rn)--(t1n)--(t2n);
            \end{scope}
            
            \begin{scope}[shift={(1,0)}]
                \node[vertex] (an) at (5,7) {};
                \node[] () at (5,7.4) {$\alpha^{i}_n$};
                \node[] () at (5,8.5) {$\mathcal{I}_1$};
                \node[vertex] (bn) at (5,2) {};
                \node[] () at (5,1.7) {$\beta^{i}_n$};
                \draw[] (an) edge [bend right] (r1);
                \draw[] (an) edge [bend left] (rn);
                \draw[] (bn) edge [bend right] (t2n);
                \draw[] (bn) edge [bend left] (t12);
            \end{scope}
    
            \begin{scope}[shift={(6,0)}]
                \node[vertex] (Ar1) at (5,6) {};
                \node[] () at (4.5,6) {$v^{i',1}_1$};
                \node[vertex] (At11) at (5,4.5) {};
                \node[] () at (4.5,4.3) {$v^{i',1}_2$};
                \node[vertex] (At12) at (5,3) {};
                \node[] () at (4.5,3) {$v^{i',1}_3$};
                \draw[] (Ar1)--(At11)--(At12);
                
                \begin{scope}[shift={(1,0)}]
                    \node[vertex] (Ar2) at (5,6) {};
                    \node[vertex] (At21) at (5,4.5) {};
                    \node[vertex] (At22) at (5,3) {};
                    \draw[] (Ar2)--(At21)--(At22);
                \end{scope}
                
                \begin{scope}[shift={(0.5,0)}]
                    \node[vertex] (Aa1) at (5,6.5) {};
                    \node[vertex] (Ab1) at (5,2.5) {};
                    \draw[] (Aa1) edge [bend right] (Ar1);
                    \draw[] (Aa1) edge [bend left] (Ar2);
                    \draw[] (Ab1) edge [bend right] (At22);
                    \draw[] (Ab1) edge [bend left] (At12);
                \end{scope}
                
                \node[] () at (6.5,4.5) {$\cdots$};
                
                \begin{scope}[shift={(2,0)}]
                    \node[vertex] (Arn) at (5,6) {};
                    \node[] () at (5.5,6) {$v^{i',n}_1$};
                    \node[vertex] (At1n) at (5,4.5) {};
                    \node[] () at (5.5,4.5) {$v^{i',n}_2$};
                    \node[vertex] (At2n) at (5,3) {};
                    \node[] () at (5.5,3) {$v^{i',n}_3$};
                    \draw[] (Arn)--(At1n)--(At2n);
                \end{scope}
                
                \begin{scope}[shift={(1,0)}]
                    \node[vertex] (Aan) at (5,7) {};
                    \node[] () at (5,8.5) {$\mathcal{I}_2$};
                    \node[] () at (5,7.4) {$\alpha^{i'}_n$};
                    \node[vertex] (Abn) at (5,2) {};
                    \node[] () at (5,1.7) {$\beta^{i'}_n$};
                    \draw[] (Aan) edge [bend right] (Ar1);
                    \draw[] (Aan) edge [bend left] (Arn);
                    \draw[] (Abn) edge [bend right] (At2n);
                    \draw[] (Abn) edge [bend left] (At12);
                \end{scope}
            \end{scope}

            \node[vertex] (g) at (9,4.5) {};
            \node[] () at (9,5) {$m_{i,i'}$};
            \draw[] (t11) edge [bend left] (g);
            \draw[] (t21) edge [bend left] (g);
            \draw[] (t1n) edge [bend left] (g);
            \draw[] (At11) edge [bend right] (g);
            \draw[] (At21) edge [bend right] (g);
            \draw[] (At1n) edge [bend right] (g);

        \end{tikzpicture}
    \caption{Adjacency gadget $\hat{A}(i,i,i',i')$.}
    \label{fig:td_lb_adj_gadget}
    \end{figure}
        
    Let graph $H$ be the adjacency gadget $\hat{A}(1,k,1,k)$ and set $\ell = \gamma \cdot (n-1) + \delta$,
    where $\gamma = 2k (2k - 1)$ and $\delta = 5k^2 - 4k$.
    Notice that it holds that $|V(H)| = (n \cdot k)^{\bO(1)}$ as well as that all copy and validation vertices have disjoint neighborhoods.
    Additionally, let $\M_{\alpha \beta}$ denote the set of pairs of $\M$ of the form $(\alpha^i_j, \beta^i_j)$.
    This concludes the construction of the instance $(H, \M, \ell)$.
    
    \begin{lemmarep}[\appsymb]\label{lem:td_lb_helper}
        $H$ has the following properties:
        \begin{itemize}
            \item $\gamma$ is equal to the number of instances of choice gadgets present in $H$,
            \item $\delta$ is equal to the number of copy and validation vertices present in $H$.        
        \end{itemize}
    \end{lemmarep}
    
    \begin{proof}
        First we will prove that for every adjacency gadget $\hat{A}(i_1,i_2,i'_1,i'_2)$ appearing in $H$,
        it holds that $i_2 - i_1 = i'_2 - i'_1 = 2^c - 1$, for some $c \in \N$.
        The statement holds for $\hat{A}(1,k,1,k)$, as well as when $i_2 - i_1 = i'_2 - i'_1 = 0$.
        Suppose that it holds for some $\hat{A}(i_1,i_2,i'_1,i'_2)$, i.e., $i_2 - i_1 = i'_2 - i'_1 = 2^c - 1 > 0$, for some $c \in \N$.
        Then, it follows that $\floor*{\frac{i_1+i_2}{2}} - i_1 = \floor*{i_2 - 2^{c-1} + 0.5} - i_1 = i_2 - i_1 - 2^{c-1} = 2^{c-1} - 1$.
        Moreover, it follows that $i_2 - \ceil*{\frac{i_1+i_2}{2}} = i_2 - \ceil*{i_1 + 2^{c-1} - 0.5} = i_2 - (i_1 + 2^{c-1}) = 2^{c-1} - 1$.
        Therefore, the stated property holds.
    
        In that case, for some $\hat{A}(i_1,i_2,i'_1,i'_2)$, in every step of the recursion,
        intervals $[i_1,i_2]$ and $[i'_1,i'_2]$ are partitioned in the middle,
        and an adjacency gadget is considered for each of the four combinations.
        In that case, starting from $\hat{A}(1,k,1,k)$, there is a single way to produce
        every adjacency gadget $\hat{A}(i_1,i_1,i_2,i_2)$, where $i_1,i_2 \in [k]$.
    
        For the first statement,
        notice that the number of instances of choice gadgets is given by the recursive formula $T_1(k) = 2k + 4T_1(k/2)$, where $T_1(1) = 2$.
        In that case, it follows that
        \[
            T_1(k) = \sum_{i=0}^{\log k} \parens*{4^i \cdot 2 \cdot \frac{k}{2^i}} =
            2k \sum_{i=0}^{\log k} 2^i =
            2k (2k - 1) =
            \gamma.
        \]
    
        For the second statement,
        notice that the number of copy plus the number of validation vertices is given by the recursive formula $T_2(k) = 4k + 4T_2(k/2)$, where $T_2(1) = 1$.
        Consequently,
        \begin{align*}        
            T_2(k) &= \sum_{i=0}^{\log k-1} \parens*{4^i \cdot 4 \cdot \frac{k}{2^i}} + 4^{\log k} =
            k^2 + 4k \sum_{i=0}^{\log k - 1} 2^i\\
            &= k^2 + 4k (k-1) =
            5k^2 - 4k = \delta,
        \end{align*}
        and the statement follows.
    \end{proof}

    \begin{lemmarep}[\appsymb]\label{lem:td_lb_td}
        It holds that $\td(H) = \bO(k)$ and $\fvs(H) = \bO(k^2)$.
    \end{lemmarep}
    
    \begin{proof}
        Let $T(\kappa)$ denote the tree-depth of $\hat{A}(i_1,i_2,i'_1,i'_2)$ in the case when $i_2 - i_1 = i'_2 - i'_1 = \kappa$.        
        First, notice that, for $i_1,i_2 \in [k]$, the tree-depth of $\hat{A}(i_1,i_1,i_2,i_2)$ is less than $8$.
        To see this, observe that by removing vertices $m_{i_1,i_2}, v^{i_1,1}_1,v^{i_1,1}_3,v^{i_2,1}_1$ and $v^{i_2,1}_3$,
        all remaining connected components are paths of at most $5$ vertices.
        Consequently, $T(1) \leq 8$.
    
        Now, consider the adjacency gadget $\hat{A}(i_1,i_2,i'_1,i'_2)$, where $i_2 - i_1 = i'_2 - i'_1 = \kappa$.
        This is comprised of adjacency gadgets
        \begin{multicols}{2}
            \begin{itemize}
                \item $\hat{A}\parens*{i_1, \floor*{\frac{i_1+i_2}{2}}, i'_1, \floor*{\frac{i'_1+i'_2}{2}}}$,
                \item $\hat{A}\parens*{i_1, \floor*{\frac{i_1+i_2}{2}}, \ceil*{\frac{i'_1+i'_2}{2}}, i'_2}$,
                \item $\hat{A}\parens*{\ceil*{\frac{i_1+i_2}{2}}, i_2, i'_1, \floor*{\frac{i'_1+i'_2}{2}}}$,
                \item $\hat{A}\parens*{\ceil*{\frac{i_1+i_2}{2}}, i_2, \ceil*{\frac{i'_1+i'_2}{2}}, i'_2}$,
            \end{itemize}
        \end{multicols}
        as well as of exactly $2\kappa$ original choice gadgets,
        each of which is connected with two copy gadgets to other instances of choice gadgets present in the adjacency gadgets.
        By removing all copy vertices $g$ (see \cref{fig:td_lb_copy_choice_gadget}),
        all the original choice gadgets as well as the adjacency gadgets are disconnected.
        Therefore, it holds that $T(\kappa) \leq 4 \kappa + T(\kappa / 2)$,
        thus, it follows that
        \[
            T(k) \leq 8 + 4 \sum_{i=0}^{\log k - 1} \frac{k}{2^i} = \bO(k).
        \]
    
        For the feedback vertex number,
        let $S \subseteq V(H)$ consist of all the copy and validation vertices.
        Due to \cref{lem:td_lb_helper}, it holds that $|S| = \delta = \bO(k^2)$,
        while $H - S$ is a collection of $\gamma$ disconnected choice gadgets.
        Let $S' = S \cup \setdef{v^{i,1}_1, v^{i,1}_3}{\hat{C}_i \text{ in } H}$,
        and notice that $H-S'$ is a collection of paths of length at most $5$,
        while $|S'| = \bO(k^2)$, thus the statement follows.
    \end{proof}

    \begin{lemmarep}[\appsymb]\label{lem:td_lb_cor1}
        If $G$ contains a $k$-clique,
        then $(H, \M, \ell)$ is a Yes instance of \mNDP.
    \end{lemmarep}
    
    \begin{proof}
        Consider $s \colon [k] \to [n]$ such that $\mathcal{V} = \setdef{v^i_{s(i)}}{i \in [k]} \subseteq V(G)$ is a $k$-clique of $G$.
        We construct a family of $\ell$ vertex-disjoint paths as follows.

        First, for every instance of $\hat{C}_i$, where $i \in [k]$, and every $j \in [2,n]$,
        we route a path from $\alpha^i_j$ to $\beta^i_j$ through the path $\hat{P}^i_j$ if $j \neq s(i)$;
        if $j = s(i)$, then we use the path $\hat{P}^i_1$ instead.
        Note that in this step we have created $\gamma \cdot (n-1)$ vertex-disjoint paths connecting terminal pairs,
        and in every gadget $\hat{C}_i$ the only unused vertices are vertices on the path $\hat{P}^i_{s(i)}$.
        
        Then, consider the adjacency gadget $\hat{A}(i,i,i',i')$, where $i, i' \in [k]$.
        For every such adjacency gadget, we take the $3$-vertex path from $v^{i,s(i)}_2$ to $v^{i',s(i')}_2$ through $m_{i,i'}$;
        note that the assumption that $\braces{v^i_{s(i)}, v^{i'}_{s(i')}} \in E(G)$ ensures that $(v^{i,s(i)}_2, v^{i',s(i')}_2) \in \M$.
        
        Lastly, let $(\mathcal{I}_1, \mathcal{I}_2, t)$ be a copy gadget,
        where $\mathcal{I}_1$, $\mathcal{I}_2$ are instances of $\hat{C}_i$ and $t \in \braces{1,2}$.
        For every such copy gadget, we take the $3$-vertex path from $v^{i,s(i)}_{t+1}$ of $\mathcal{I}_1$ to
        $v^{i,s(i)}_1$ of $\mathcal{I}_2$ through the copy vertex $g$ connecting them.

        Since the neighborhoods of all the copy and validation vertices are disjoint,
        this is a valid routing, and it follows that exactly $\gamma \cdot (n-1) + \delta = \ell$ demands are routed.
    \end{proof}

    \begin{lemmarep}[\appsymb]\label{lem:td_lb_cor2}
        If $(H, \M, \ell)$ is a Yes instance of \mNDP,
        then $G$ contains a $k$-clique.
    \end{lemmarep}
    
    \begin{proof}
        Let $\mathcal{P} = \braces{P_1, \ldots, P_\ell}$ be a set of $\ell$ vertex-disjoint paths connecting terminal pairs of~$\M$ in~$H$.
        Assume without loss of generality that said paths are simple,
        as well as that the only edges among vertices of the path appear between consecutive vertices.
        Moreover, let $\mathcal{P}_{\alpha \beta} \subseteq \mathcal{P}$ contain all the paths of $\mathcal{P}$
        that route demands of $\M_{\alpha \beta}$.
        Set $S$ to be the set consisting of all the copy and validation vertices in $H$, where $|S| = \delta$ due to \cref{lem:td_lb_helper}.

        Notice that at most $\delta$ paths of $\mathcal{P}$ contain vertices of $S$, consequently at least $\ell - \delta = \gamma \cdot (n-1)$ paths
        route demands using only vertices of $H-S$.
        Moreover, $H-S$ is a collection of $\gamma$ disconnected choice gadgets, each of which can be used to route at most $n-1$ demands.
        Consequently, it follows that exactly $\delta$ paths of $\mathcal{P}$ contain vertices of $S$,
        while in each choice gadget, exactly $n-1$ demands are routed using vertices of $H-S$,
        thus all the demands of~$\M_{\alpha \beta}$ are routed.
        Let $P_i$, for $i \in [\delta]$, denote the paths of $\mathcal{P} \setminus \mathcal{P}_{\alpha \beta}$, each of which contains exactly one vertex of $S$.

        Notice that by routing all $n-1$ demands $(\alpha^i_j, \beta^i_j)$ in a choice gadget $\hat{C}_i$, where $j \in [2,n]$,
        only the vertices of $\hat{P}^i_j$ remain unused, for some $j \in [n]$.
        To see why this property holds, notice that the shortest path connecting the terminals of such a demand is of length $5$, including said terminals.
        Therefore, at most $3$ vertices of a choice gadget may remain unused.
        However, if said vertices belonged to more than a single path $\hat{P}^i_j$, then at most $n-2$ demands can be routed,
        which is a contradiction.
        Since every such demand can be routed via either $\hat{P}^i_j$ or $\hat{P}^i_1$,
        there exist $n-1$ paths of $\mathcal{P}_{\alpha \beta}$ that route all demands associated with $\hat{C}_i$ and does not use the vertices of $\hat{P}^i_j$,
        for some $j \in [n]$.

        Next, we show that every $P_i$, for $i \in [\delta]$, involves exactly $3$ vertices,
        with one being a vertex of $S$ and the other two being its neighbors.
        Fix one such $i$, and let $s$ denote the copy or validation vertex of $S$ appearing in $P_i$,
        while $\mathcal{I}_1$ and $\mathcal{I}_2$ denote the choice gadgets that have vertices adjacent to $s$.
        Set $S' = S \setminus \braces{s}$.
        Then, it holds that $P_i$ routes one demand of $\M \setminus \M_{\alpha \beta}$, using the vertices of the graph $H - S'$.
        In that case, one endpoint of such a demand can only be a vertex of $\mathcal{I}_1$ and the other of $\mathcal{I}_2$,
        and due to the construction of $(H,\M,\ell)$, it holds that for every such demand, both its endpoints are neighbors of $s$.

        Next we show that, for $i \in [k]$, for every instance of $\hat{C}_i$
        it holds that the only vertices not appearing in paths of $\mathcal{P}_{\alpha \beta}$
        are those of $\hat{P}^{i,j}_{s(i)}$,
        for some function $s \colon [k] \to [n]$.
        Assume there exists a copy gadget $(\mathcal{I}_1, \mathcal{I}_2, t)$,
        where $\mathcal{I}_1$ and $\mathcal{I}_2$ are instances of $\hat{C}_i$,
        $t \in \braces{1,2}$ and $s$ is the corresponding copy vertex.
        Moreover, let $P \in \mathcal{P}$ such that $s \in P$.
        Then, there exists $s(i) \in [n]$ such that the vertices of $\hat{P}^i_{s(i)}$ are not used to route demands in $\M_{\alpha \beta}$,
        for both $\mathcal{I}_1$ and $\mathcal{I}_2$,
        as otherwise no demand of $\M \setminus \M_{\alpha \beta}$ can be routed in the graph induced by $\mathcal{I}_1$, $\mathcal{I}_2$ and $s$.
        Furthermore, notice that for every $s_1, s_2 \in S$, $N(s_1) \neq N(s_2)$,
        while for every instance $\mathcal{I}' \neq \mathcal{I}$ of $\hat{C}_i$,
        there is a sequence of copy gadgets $(\mathcal{I}, \mathcal{I}_1,t_1), \ldots, (\mathcal{I}_p, \mathcal{I}', t_{p+1})$,
        where $\mathcal{I}$ denotes the original choice gadget $\hat{C}_i$ in $\hat{A}(1,k,1,k)$.

        Let $\mathcal{V} = \setdef{v^i_{s(i)}}{i \in [k]} \subseteq V(G)$,
        where $|\mathcal{V} \cap V_i| = 1$, for all $i \in [k]$.
        We will prove that $\mathcal{V}$ is a clique.
        Let $v^{i_1}_{s(i_1)}, v^{i_2}_{s(i_2)} \in \mathcal{V}$.
        Consider the adjacency gadget $\hat{A}(i_1,i_1,i_2,i_2)$ and let $P_i \in \mathcal{P}$ be the path of $\mathcal{P}$ that contains $m_{i_1, i_2}$.
        It holds that $P_i$ is comprised of $m_{i_1, i_2}$, as well as two of its neighbors, one in $\mathcal{I}_1$ of $\hat{C}_{i_1}$ and one in $\mathcal{I}_2$ of $\hat{C}_{i_2}$.
        Since the only neighbors of $m_{i_1,i_2}$ that are not used by paths in $\mathcal{P}_{\alpha \beta}$ are $v^{i_1,s(i_1)}_2$ and $v^{i_2,s(i_2)}_2$,
        we infer that $(v^{i_1,s(i_1)}_2, v^{i_2,s(i_2)}_2) \in \M$ and, consequently, $\braces{v^{i_1}_{s(i_1)}, v^{i_2}_{s(i_2)}} \in E^{i_1,i_2}$.
        Since this holds for any two vertices belonging to~$\mathcal{V}$, it follows that $G$ has a $k$-clique.
    \end{proof}

    Therefore, in polynomial time, we can construct a graph $H$,
    of tree-depth $\td = \bO(k)$ and $\fvs = \bO(k^2)$ due to \cref{lem:td_lb_td},
    as well as a set of pairs $\M$,
    such that, due to \cref{lem:td_lb_cor1,lem:td_lb_cor2},
    deciding whether at least $\ell$ pairs of $\M$ can be routed 
    is equivalent to deciding whether $G$ has a $k$-clique.
\end{proof}

\section{Conclusion}\label{sec:conclusion}
In this work we examine in depth \textsc{Maximum Node-Disjoint Paths}
under the perspective of parameterized complexity,
painting a complete picture regarding its tractability under most standard parameterizations.
We additionally employ the use of approximation and present various (in)approximability results,
further enhancing our understanding of the problem in the setting of parameterized approximation.

As a direction for future work,
we remark that although \cref{thm:inapx_pw} excludes the existence of an approximation scheme,
a constant-factor approximation algorithm running in time $f(\tw) n^{\bO(1)}$ remains possible.
Regarding the FPT algorithms of \cref{sec:tractability},
all but the one of \cref{thm:fes} rely on providing an upper bound on the number of vertices
involved in an optimal solution,
leading in some cases to even double-exponential running times.
The optimality of those running times is thus a natural open question.
Lastly, it is unknown whether the problem is FPT parameterized by the cutwidth or the cluster vertex deletion number of the input graph.

\bibliography{bibliography}


\end{document}